\documentclass[journal]{IEEEtran}

\usepackage{amsfonts,amsmath,amssymb}
\usepackage{dsfont}
\usepackage{cite}
\usepackage{graphicx}
 \usepackage{booktabs}
\usepackage[pdf]{pstricks}

\usepackage{bm}
\usepackage{bbding}
\usepackage{amsmath}
\usepackage{enumerate}
\usepackage{multirow}
\usepackage[ruled]{algorithm2e}
\usepackage{algorithmicx}

\usepackage{algpseudocode}
\usepackage{algcompatible}
\usepackage{subcaption}
\usepackage{array}
\usepackage{slashbox}
\usepackage{CJK}
\usepackage{graphicx}

\usepackage{url}
\usepackage{amsthm}
\newtheorem{theorem}{Theorem}

\newcommand{\tabincell}[2]{\begin{tabular}{@{}#1@{}}#2\end{tabular}}%
 \newtheorem{lemma}{Lemma}
  \newtheorem{proposition}{Proposition}

  \newtheorem{remark}{Remark}

  \newcommand{\figref}[1]{\figurename~\ref{#1}}

\DeclareMathOperator*{\argmax}{arg\,max}

\DeclareMathOperator{\diag}{diag}

\begin{document}
\title{Joint Beam Training and Positioning for Intelligent Reflecting Surfaces Assisted Millimeter Wave Communications}
\author{Wei~Wang,~\IEEEmembership{Member,~IEEE}, and~Wei~Zhang,~\IEEEmembership{Fellow,~IEEE}
\thanks{
W. Wang and W. Zhang are with the School of Electrical Engineering and Telecommunications, The University of New South Wales, Sydney, Australia (e-mail: wei.wang@unsw.edu.au; wzhang@ee.unsw.edu.au).}
}

\maketitle

\begin{abstract}
Intelligent reflecting surface (IRS) offers a cost-effective solution to link blockage problem in mmWave communications, and the prerequisite of which is the  accurate estimation of (1) the optimal beams for base station/access point (BS/AP) and mobile terminal (MT), (2) the optimal reflection patterns for IRSs, and (3) link blockage.
In this paper, we carry out beam training designs for IRSs assisted mmWave communications to estimate the aforementioned parameters.  To acquire the optimal beams and reflection patterns, we firstly perform random beamforming and maximum likelihood estimation to estimate angle of arrival (AoA) and angle of departure (AoD) of the line of sight (LoS) path between BS/AP (or IRSs) and MT. Then, with the estimated AoDs, we propose an iterative positioning algorithm that achieves centimeter-level positioning accuracy. The obtained location information is not only a fringe benefit but also enables us to  cross verify and enhance the estimation of AoA and AoD, and  it also facilitates the estimation of blockage indicator. Numerical results show the superiority of our proposed beam training scheme and verify the performance gain brought by location information.
\end{abstract}


\maketitle

\section{Introduction}
Millimeter-wave (mmWave) band, ranging from 30GHz to 300GHz, has attracted great interests from both academia and industry for its abundant spectrum resources \cite{Mmwave0,Mmwave2}. The Wi-Fi standard IEEE 802.11ad runs on the 60GHz (V band) spectrum with data transfer rates of up to 7 Gbit/s \cite{nitsche2014ieee, sulyman2014radio}. In 3GPP Release 15,  24.25-29.5GHz and 37-43.5GHz, as the most promising frequencies for the early deployment of 5G millimeter wave systems, are specified based on a time-division duplexing (TDD) access scheme \cite{3GPP1}. The millimeter scale wavelength, on one hand, renders massive antennas integratable on an antenna array with portable size \cite{AIP1}, and, on the other hand, results in severe free space path loss especially for non-line-of-sight (NLoS) paths. Directional transmission enabled by beamforming techniques is an energy efficient transmission solution to compensate for the path loss in mmWave communications \cite{raghavan2016directional}. By properly adjusting the phase shifts of each antenna elements, it concentrates the emitted energy in a narrow beam between transmitter and receiver. However, the directional link is easily blocked by obstacles like human bodies, walls, and furniture, attributed to the millimeter scale wavelength \cite{bai2014coverage}. Once LoS path is blocked, it is highly possible that the blocked link cannot be restored no matter how the beam direction is adjusted, as the NLoS paths are not strong enough to serve as a qualified alternative link. Channel measurement campaigns reveal that power of the LoS component is about 13dB higher than the sum of power of NLoS components \cite{muhi2010modelling}. Therefore, blockage is the biggest hindrance to the large scale applications of mmWave band in mobile communication systems.

Recently, intelligent reflecting surface (IRS) \cite{wang2019intelligent, yang2019intelligent,IRSQU,wang2019channel}, a.k.a. reconfigurable intelligent surface (RIS) \cite{basar2019wireless,huang2019reconfigurable}, large intelligent surface (LIS) \cite{taha2019enabling}, passive (intelligent) reflectors/mirrors \cite{huang2018achievable,khawaja2019coverage,tan2018enabling}, or programmable metasurface \cite{tang2019wireless,zhao2018programmable,tang2019programmable}, is proposed as an energy-effective and cost-effective hardware structure for future wireless communications. IRS is essentially a new type of electromagnetic surface structure which is typically designed by deliberately arranging a set of sophisticated passive scatterers or apertures in a regular array to achieve the desired ability for guiding and controlling the flow of electromagnetic waves \cite{holloway2012overview}. Current applications of IRS to wireless communications can be categorized into two types, namely IRS modulator and IRS ``relay".  In \cite{tang2019wireless,zhao2018programmable,tang2019programmable}, amplitude/phase modulations over IRS are investigated. Through controlling the reflection coefficient of IRS, the incident carrier wave from a feed antenna can be digitally modulated without requiring high-performance radio frequency (RF) chains. A more extensive application of IRS is IRS ``relay", in which the radiated power from BS/AP towards IRS is reflected to MT via intelligently managing the phase shifters on IRS \cite{basar2019wireless,wang2019channel,huang2019reconfigurable,wang2019intelligent, yang2019intelligent,IRSQU,taha2019enabling,tan2018enabling,huang2018achievable,khawaja2019coverage}. It is noteworthy that the rationale behind IRS ``relay" and conventional amplify-and-forward (AF) relay is significantly different. AF relay firstly receives signal and then re-generates and re-transmits signal. In contrast, IRS only reflects the ambient RF signals as a passive array and bypasses conventional RF modules such as power amplifier, filters, and ADC/DAC \cite{yang2019intelligent}. Hence, IRS ``relay" incurs no additional power consumption and is free from thermal noise introduced by RF modules. In this sense, IRS can be regarded as a smart ``mirror" that enables us to change the paradigm of wireless communications from adjusting to wireless channel to changing wireless channel \cite{basar2019wireless,yang2016programmable}. As an active way to make wireless channel better, IRS ``relay" assisted wireless communications have attracted great interests from researchers.  In \cite{wang2019intelligent}, IRS is applied to mmWave communications to provide effective reflected paths and thus enhance signal coverage.  In \cite{IRSQU,huang2019reconfigurable,huang2018achievable}, joint optimization of the transmit beamforming by active antenna array at the BS/AP and reflect beamforming by passive phase shifters at the IRS is carried out.  In \cite{khawaja2019coverage}, empirical studies are performed to analyze the capability of signal coverage enhancement for IRSs assisted mmWave MIMO at 28GHz.  In \cite{tan2018enabling},  the reconfigurable 60GHz IRS is designed, implemented and deployed to strengthen mmWave connections for indoor networks threatened by blockage. The objective of the work is to validate IRS's capability to address link blockage problem in mmWave communications, and beam training design is not investigated.
Although extensive analytical and empirical studies have been done on IRSs assisted wireless communications in the aforementioned literature, these work either assume the availability of channel state information (CSI) or accurate measurement of BS/AP, MT and IRS's position and direction.

In \cite{yang2019intelligent}, a practical transmission protocol and channel estimation are firstly proposed for an IRS-assisted orthogonal frequency division multiplexing (OFDM) system under frequency-selective channels. In \cite{wang2019channel}, by exploiting the channel correlation among different users, a channel estimation scheme with reduced training overhead is proposed. Specifically, with a typical user's reflection channel vector, estimation of the other users' reflection channel vector can be simplified as the estimation of a multiplicative coefficient.  However, the aforementioned designs were performed in non-mmWave frequency band, and the direct application of them to mmWave communications will fail to utilize the sparse nature of mmWave channel. In \cite{taha2019enabling}, to facilitate channel estimation of IRSs assisted link over mmWave band or LoS dominated sub-6GHz band, an upgrade of IRS's structure is proposed to add a small number of channel sensors to sense and process incident signal. Although \cite{taha2019enabling} is intended to mmWave band, the proposed compressive sensing and deep learning algorithms are incompatible to current structure of IRS which is without channel sensors.
{ In \cite{wang2020compressed}, cascade channel estimation of the BS/AP-IRS-MT link in mmWave band is firstly converted into a sparse signal recovery
problem and then solved via conventional compressed sensing methods. However, \cite{wang2020compressed} is based on a strong assumption that AoA and AoD parameters lie on the discretized grid. In \cite{HighResolution}, a two-step  channel estimation protocol is proposed for the cascaded BS/AP-IRS-MT link in mmWave band, which includes hierarchical beamforming  and high resolution sparse channel estimation. As the selection of fine beam set in hierarchical beamforming is fully dependent on the training results of the wide beams in the previous layer, hierarchical beamforming requires interactions between BS/AP and MT. Thus, the extension of the proposed scheme from single user scenario to multi-user scenario might be costly in training overhead. Besides, as IRS is primarily used in mmWave communications to combat blockage, estimation of blockage in both BS/AP-MT link and BS/AP-IRS-MT link is essential for IRSs assisted mmWave communications, while  \cite{taha2019enabling, wang2020compressed, HighResolution} all neglect blockage effects in their designs.
}

Due to the deployment of multiple IRSs, beam training of IRSs assisted mmWave communications requires much heavier training overhead than traditional mmWave communications. Also, as the purpose of IRSs is to combat blockage and expand coverage, an accurate estimation of  blockage is essential to beam selection by BS/AP. In addition, the lack of RF chains results in the inability of IRSs to sense signal, which further complicates beam training for the paths assisted by IRSs. These  three features jointly render traditional beam training methods \cite{xiao2016hierarchical,Mmwave5} incompetent in IRSs assisted mmWave communications.
Despite the aforementioned new challenges of integrating IRSs to mmWave communications, a notable advantage is that the estimation of path parameters, e.g., AoA/AoD and blockage indicator, can be cross verified, thanks to the relatively large number of deployed IRSs. Specifically, three accurate estimates of AoA/AoD, associated with other essential information, e.g., direction of arrays,  can yield the location of MT, and the location of MT will in turn reproduce the path parameters. In this way, the path parameters of IRSs assisted mmWave MIMO can be  enhanced according to their geometric relationship. To estimate the channel parameters of IRSs assisted mmWave communications, we have made the following contributions in this paper:
\begin{itemize}
\item We propose a flexible beam training method for IRSs assisted mmWave MIMO by breaking it down into several mathematically equivalent sub-problems, and we further perform random beamforming and maximum likelihood (ML) estimation to jointly estimate AoA and AoD of the dominant path in each sub-problem. The proposed scheme does not require feedback from MT at training stage, and thus can be performed in a broadcasting manner. Hence, the required training overhead will not increase over MT number.
\item We prove the uniqueness of the AoA and AoD estimated by beam training with random beamforming. We further study the impact of training length, and we prove that larger training length almost surely results in smaller pairwise error probability of AoA, AoD pair.
\item By sorting the reliability of the estimated AoA, AoD pairs, we propose an iterative positioning algorithm to estimate the location of MT, and, through numerical analysis, we show that the algorithm achieves centimeter-level positioning accuracy.
\item With the estimated position of MT, we propose to cross verify and enhance the estimation of path parameters, i.e., AoA and AoD, according to their geometric relationship. We further propose an accurate method of blockage estimation by comparing the ML estimate of pathloss and MT position based estimate of pathloss.
\end{itemize}
Numerical results show the superiority of our proposed beam training scheme and verify the performance gain brought by location information.

The rest of the paper is organized as follows. Section II introduces the system model. In Section III, we break down the beam training design of IRSs assisted mmWave communications. In Section IV, we propose beam training with random beamforming, and specifically we estimate path parameters and study the feasibility of the scheme. In Section V, we study the interplay between  positioning and beam training. In Section VI, numerical results are presented. Finally, in Section VII, we draw the conclusion.

{\em{Notations:\quad}} Column vectors (matrices) are denoted by bold-face lower (upper) case letters, $\mathbf{x}(n)$ denotes the $n$-th element in the vector $\mathbf{x}$, $(\cdot)^*$, $(\cdot)^T$ and $(\cdot)^{H}$  represent conjugate, transpose and  conjugate transpose operation, respectively,  $||\cdot||$ denotes the Frobenius norm of a vector or a matrix,  $\odot$ is Hadamard product. Subtraction and addition of the cosine AoAs/AoDs are defined as $\theta \ominus \phi  \triangleq (\theta - \phi+1) \mod 2 -1$ and $\theta \oplus \phi  \triangleq (\theta + \phi+1) \mod 2 -1$ to guarantee the result is within the range $[-1,1)$.


\section{System Model}

Consider a communication link between the BS/AP and an MT operating in mmWave band, where both ends adopt uniform linear array (ULA) antenna structure. To reduce wireless link blockage rate and thus guarantee the reliable linkage between BS and MT, a number of IRSs are deployed in the cell as shown in \figref{RIS1}, and BS/AP is able to control IRSs via cable or lower frequency radio link.

\begin{figure}[tp]{
\begin{center}{\includegraphics[ height=3.6cm]{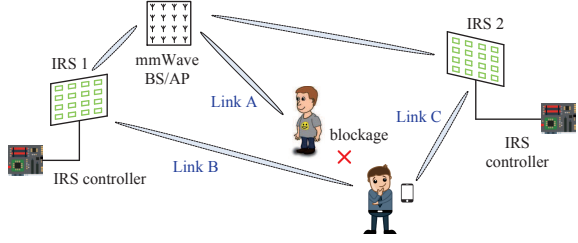}}
\caption{Illustration of IRSs assisted mmWave communications}\label{RIS1}
\end{center}}
\end{figure}

The channel response between BS/AP and MT without the assistance of IRSs is represented as \cite{Mmwave5}
\begin{align}
 \mathbf{H}_{BM} =  &\zeta_{LoS}  \delta_{1} \mathbf{a}_{M}(\theta_{BM,1}) \mathbf{a}^H_{B}(\phi_{BM,1})  + \notag \\
 &\sum_{l=2}^{L}\delta_{l} \mathbf{a}_{M}(\theta_{BM,l}) \mathbf{a}^H_{B}(\phi_{BM,l}) \label{Channel}
\end{align}
where $\zeta_{LoS}\in \{0, 1 \}$ is the indicator of blockage of the LoS path, and $\delta_{l}$, $\theta_{BM,l}$ and $\phi_{BM,l}$ are channel gain, cosine of AoA, and cosine of  AoD  of the $l$-th path, respectively. The parameters $(\zeta_{LoS}, \delta_{1}, \theta_{BM, 1}, \phi_{BM,1})$ characterize LoS path, which are of particular interest to us in mmWave communications. According to \cite{basar2019wireless}, the path gain of LoS is $\delta_{1} =  \frac{\lambda e^{-j 2\pi d_{BM}}}{4\pi d_{BM}}$,
where $\lambda$ is the wavelength, and $d_{BM}$ is the distance between BS and MT.
Further, the steering vectors are given by
\begin{align}
\mathbf{a}_{M}(\theta_{BM, l})  =  [1,\; e^{j \pi 1 \theta_{BM, l}},\;\cdots,e^{j \pi (N_M-1) \theta_{BM, l}}]^T \notag \\
\mathbf{a}_{B}(\phi_{BM, l}) =  [1,\; e^{j \pi 1 \phi_{BM, l}},\;\cdots,e^{j \pi (N_B-1) \phi_{BM, l}}]^T \notag
\end{align}
where $N_B$ is the number of antennas of BS/AP, $N_M$ is the number of antennas of MT.

We also assume that IRSs adopt ULA antenna structure. Thus, the channel response of the reflected path from BS to MT assisted (reflected) by  the  $i$-th IRS is
\begin{align}
  \mathbf{H}_{BR_iM}=& \zeta_{VLoS,i} \bar \delta_{B R_i M}  \mathbf{a}_{M}(\theta_{R_iM}) \mathbf{a}_{R_i}^H(\phi_{R_iM}) \notag \\
    & \diag \{\bar{\mathbf{g}}_i\} \mathbf{a}_{R_i}(\theta_{BR_i}) \mathbf{a}_{B}^H(\phi_{BR_i})  \notag \\
  = &\zeta_{VLoS,i} \delta_{B R_i M} (\bar{\mathbf{g}}_i)  \mathbf{a}_{M}(\theta_{R_iM}) \mathbf{a}_{B}^H(\phi_{BR_i})
\end{align}
{ where $\bar{\mathbf{g}}_i$ is the reflection vector that determines the reflection pattern of the $i$-th IRS}, $\zeta_{VLoS,i} \in \{0, 1 \} $ is the indicator of blockage of the  path reflected by the $i$-th IRS and $\bar \delta_{BR_iM} = \frac{\sqrt{\xi} \lambda e^{-j 2\pi (d_{BR_i} + d_{R_i M}) }}{4\pi(d_{BR_i} + d_{R_i M})}$ \cite{basar2019wireless}, in which
$\xi$ is reflection loss, $d_{BR_i}$ is the distance between BS and the $i$-th IRS, $d_{R_i M}$ is the distance between the $i$-th IRS and MT. The equivalent path gain of the IRS reflected path can be written as
\begin{align} \label{PathGain2}
\delta_{B R_i M} (\bar{\mathbf{g}}_i)  & \triangleq \bar \delta_{B R_i M}  \mathbf{a}_{R_i}^H(\phi_{R_iM})\diag \{\bar{\mathbf{g}}_i\}   \mathbf{a}_{R_i}(\theta_{BR_i})  \notag \\
& =  \bar \delta_{B R_i M}\mathbf{a}_{R_i}^H(\phi_{R_iM} \ominus \theta_{BR_i}) \bar{\mathbf{g}}_i
\end{align}
The steering vector $\mathbf{a}_{R_i}(\phi_{R_iM}) $ is given by
\begin{align}
\mathbf{a}_{R_i}(\phi_{R_iM})  =  [1,\; e^{j \pi 1 \phi_{R_iM} },\;\cdots,e^{j \pi (N_{R_i} -1) \phi_{R_iM}}]^T
\end{align}
where $N_{R_i}$ is the number of passive reflectors of the $i$-th IRS. { Based on \eqref{PathGain2}, the optimal reflection coefficient vector that maximizes effective received power is $\bar{\mathbf{g}}_i^\star = \mathbf{a}_{R_i}(\phi_{R_iM} \ominus \theta_{BR_i})$}.

Hence, the channel response between BS and MT with the assistance of $N_{IRS}$ IRSs is represented as
\begin{align}
 &\mathbf{H}   =  \mathbf{H}_{BM} +  \sum_{i=1}^{N_{IRS}} \gamma_i \mathbf{H}_{BR_iM} =\notag \\
   &\underbrace{\zeta_{LoS}\delta_{1} \mathbf{a}_{M}(\theta_{MB,1}) \mathbf{a}^H_{B}(\phi_{MB,1})}_{LoS \; component} + \underbrace{\sum_{l=2}^{L}\delta_{l} \mathbf{a}_{M}(\theta_{MB,l}) \mathbf{a}^H_{B}(\phi_{MB,l})}_{NLoS \; component}  \notag \\
 &\;\; + \underbrace{\sum_{i=1}^{N_{IRS}} \gamma_i  \zeta_{VLoS,i}\delta_{B R_i M} (\bar{\mathbf{g}}_i)  \mathbf{a}_{M}(\theta_{R_iM}) \mathbf{a}_{B}^H(\phi_{BR_i})}_{VLoS \; component} \label{ChannelResponse}
\end{align}
where
\begin{align}
\gamma_i = \left\{\begin{array}{cc}
             1, & {\rm when\;the\;} i{\rm th\; IRS\; is\; activated}  \\
             0, & {\;\;\;\;\rm when\; the\;} i{\rm th\; IRS\; is\; deactivated}
           \end{array} \right.
   \notag
\end{align}
indicates the activation status of the $i$-th IRS and $\gamma_i$ can be configured by BS/AP.

When the reflection pattern of the vector $\bar{\mathbf{g}}_i$ is omnidirectional, IRS works as a scatterer that diffuses the energy radiated from BS. When $\bar{\mathbf{g}}_i^\star = \mathbf{a}_R(\phi_{R_iM} \ominus \theta_{BR_i})$, IRS works as a ``mirror'' that builds a virtual LoS (VLoS) path between BS and MT, and thus the energy from BS will be concentrated on MT, and $\phi_{R_iM} \ominus \theta_{  BR_i}$ is termed as \emph{the optimal reflection angle} of the $i$-th VLoS path. We can categorize channel components of ${\mathbf{H}} $ into three types as  in Eq. \eqref{ChannelResponse}, namely \emph{LoS path component}, \emph{VLoS path component}, and \emph{NLoS path component}. LoS path component is the direct path between BS and MT,  VLoS path component consists of the paths between BS and MT reflected by IRSs, and NLoS path component consists of  the paths between BS and MT reflected by scatters, e.g., walls, human bodies, and etc.

As NLoS path component usually varies fast and its weight to the channel is marginal especially in mmWave band, we are more interested in LoS path and VLoS paths. Hence, we intend to estimate (1) the optimal reflection angle $\phi_{R_iM} \ominus \theta_{BR_i}$ of IRSs and (2) the path parameters $(\zeta_{BM,1}, \delta_{BM,1}, \theta_{BM,1}, \phi_{BM,1})$ of the LoS path and $(\zeta_{BR_iM}, \delta_{BR_iM}({\bar{\mathbf{g}}}^\star_i), \theta_{R_iM},\phi_{BR_i})$ of the VLoS paths through beam training and location information aided parameter enhancement.

\section{Framework of Joint Beam Training and Positioning}

For conventional mmWave communications,  training overhead  can be significantly reduced by exploiting the sparse nature of mmWave channel \cite{Mmwave5,OTHOR}. However, with the assistance of IRSs, the sparse channel of mmWave band is artificially converted into rich scattering channel. The increased scattering effect, together with the unknown optimal reflection angle, jointly complicates the process of beam training. To make the over-complicated problem tractable, we propose to { break} down beam training of IRSs assisted mmWave MIMO into two sub-problems, and we further show that the two sub-problems are mathematically equivalent. { Then,  we propose a protocol for joint beam training and positioning which well accommodates multi-user scenario.}

\subsection{Breakdown of Beam Training for IRSs Assisted MmWave MIMO}

\begin{figure}[tp]{
\begin{center}{\includegraphics[ height=4.5cm]{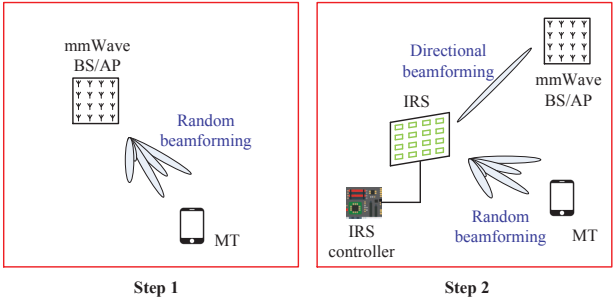}}
\caption{Two steps of  beam training with random beamforming in IRSs assisted mmWave communications}\label{RIS3}
\end{center}}
\end{figure}

At first, it is noteworthy that AoA/AoD of the LoS path between IRSs and BS/AP can be accurately pre-measured, since both IRSs and BS/AP are pre-configured. Thus,  $\theta_{B R_i}$ and $\phi_{B R_i}$ are used as prior knowledge hereafter. Then, beam training of IRSs assisted mmWave MIMO is carried out in the following two steps  as  illustrated in \figref{RIS3}.

 {\textbf{Step 1.} \underline{De-activate all the IRSs, and estimate the parameters} \underline{$(\delta_{BM,1}, \theta_{BM,1}, \phi_{BM,1})$ of LoS path}}


To estimate the parameters, measures of  channel are collected via Tx/Rx random beamforming in BS/AP side and MT side, i.e.,
\begin{align} \label{MeasureBS}
 y = &\sqrt{P_{Tx}}\mathbf{m}^H   \mathbf{H}_{BM} \mathbf{f}s  + \mathbf{m}^H\bar{\mathbf{w}} \notag \\
=  & \sqrt{P_{Tx}} \zeta_{LoS} \delta_{BM,1}  \mathbf{m}^H\mathbf{a}_{M}(\theta_{BM,1}) \mathbf{a}_{B}^H(\phi_{BM,1})\mathbf{f} + \notag \\
 & \underbrace{\sum_{l=2}^{L}\sqrt{P_{Tx}}\delta_{BM,l} \mathbf{m}^H\mathbf{a}_{M}(\theta_{BM, l}) \mathbf{a}_{B}^H(\phi_{BM,l})\mathbf{f}}_{\nu} + \mathbf{m}^H\bar{\mathbf{w}}
\end{align}
where $P_{Tx}$ is transmit power, $ \bar{\mathbf{w}} \sim {\cal {CN}}(\mathbf{0}, \sigma^2_{\bar{\mathbf{w}}} \mathbf{I}_{N_M})$ is the zero-mean complex Gaussian additive noise, $s=1$ is the pilot signal sent by the user,
$ \mathbf{f} $ and $\mathbf{m}$  are  transmit random beamforming vector at BS/AP side and receive random beamforming vector at MT side\footnotemark, respectively, and
the entries of $ \mathbf{f} $ and $\mathbf{m}$ are phase-only complex variables with invariable amplitude \cite{Mmwave6}, i.e.,
\begin{align}
  \mathbf{f} = \frac{1}{\sqrt{N_B}} \left(e^{j\pi \varrho_{1}}, e^{j\pi \varrho_{2}}, \cdots, e^{j\pi \varrho_{N_B}} \right)^T \notag \\
  \mathbf{m} = \frac{1}{\sqrt{N_M}} \left(e^{j\pi \sigma_{1}}, e^{j\pi \sigma_{2}}, \cdots, e^{j\pi \sigma_{N_M}\emph{}} \right)^T \notag
\end{align}
 $\varrho_{n_B}$ is the phase shift value of the $n_B$-th analog phase shifter in BS/AP side, $\sigma_{n_M}$  is the phase shift value of the $n_M$-th analog phase shifter in MT side.

\footnotetext{A good random beamforming codebook can be derived offline by high performance computers, and they will be pre-configured in BS/AP, IRS and MT side.}


As NLoS paths are much weaker than LoS path in mmWave band, i.e.,  $ \delta_{BM,l} (l=2,\cdots, L)$ are small compared to $\delta_{BM,1}$,  we are very less likely to build an effective communication link via NLoS paths. Hence, the AoA, AoD pair that we are interested in is merely $(\zeta_{LoS}, \delta_{BM,1}, \theta_{BM,1}, \phi_{BM,1})$, and the term $\nu$ will be treated as interference. Considering the small scale and randomness of $\delta_{BM,l} (l=2,\cdots, L)$, we assume that $\nu$ follows complex Gaussian distribution for the simplicity of analysis\footnotemark.  Then, the beam training problem for IRSs assisted mmWave MIMO communications is formulated as the estimation of  $(\zeta_{LoS}, \delta_{BM,1}, \theta_{BM,1}, \phi_{BM,1})$ from the following received signal
\begin{align}
 y =  \sqrt{P_{Tx}}{\zeta_{LoS} \delta_{BM,1}   \mathbf{m}^H \mathbf{a}_{M}(\theta_{BM,1}) \mathbf{a}^H_{B}(\phi_{BM,1})  \mathbf{f}} + \nu+  \mathbf{m}^H \bar{\mathbf{w}}
\end{align}
Adding the subscript $n$ to $y$ to denote the received signal in the $n$-th time slot, we have
\begin{align}
      y_{n} & = \sqrt{P_{Tx}} \zeta_{LoS}\delta_{BM,1}\mathbf{m}_n^H \mathbf{a}_{M}(\theta_{BM,1}) \mathbf{a}^H_{B}(\phi_{BM,1})\mathbf{f}_n \notag \\
   & \qquad  + \nu_n  + \mathbf{m}_n^H \bar{\mathbf{w}}_n \notag \\
& =  \sqrt{P_{Tx}}\zeta_{LoS}\delta_{BM,1} (\mathbf{f}_n^T \otimes \mathbf{m}_n^H) \mathbf{b} (\theta_{BM,1}, \phi_{BM,1}) \notag \\
&\qquad + \nu_n+ \mathbf{m}_n^H\bar{\mathbf{w}}_n \notag
\end{align}
where $\mathbf{b} (\theta_{BM,1}, \phi_{BM,1}) \triangleq  vec( \mathbf{a}_{M}(\theta_{BM,1}) \mathbf{a}^H_{B}(\phi_{BM,1}))$.

\footnotetext{Although we assume that $\nu$ follows Gaussian distribution in theoretical analysis, the channel model to be applied in numerical simulations still considers NLoS components as in Eq. \eqref{Channel}. }

To estimate AoA and AoD, $N$ channel measurements are to be collected and concatenated, and its vector form is derived as
\begin{align}  \label{Sysmodel}
\mathbf{y} =  \sqrt{P_{Tx}} \zeta_{LoS}\delta_{BM,1} \mathbf{D} \mathbf{b} (\theta_{BM,1}, \phi_{BM,1}) + \underbrace{\boldsymbol{\nu}+ \mathbf{w}}_{\mathbf{n}}
\end{align}
where
\begin{align}
\mathbf{y} & = \left[y_{1},\; y_{2},\;\cdots, y_{N} \right]^T \notag\\
\mathbf{D} &= \left[\mathbf{f}_1 \otimes \mathbf{m}_1^*,\; \mathbf{f}_2  \otimes \mathbf{m}_2^*,\; \cdots, \;\mathbf{f}_N  \otimes \mathbf{m}_N^*\right]^T \notag\\
\boldsymbol{\nu} &= \left[ {\nu}_1,\;  {\nu}_2,\; \cdots, \;  {\nu}_N \right]^T  \sim   {\cal {CN}}(\mathbf{0}, \sigma^2_{\boldsymbol{\nu}} \mathbf{I}_{N}) \notag\\
\mathbf{w} &= \left[\mathbf{m}_1^H\bar{\mathbf{w}}_1,\; \mathbf{m}_2^H\bar{\mathbf{w}}_2,\; \cdots, \;\mathbf{m}_N^H \bar{\mathbf{w}}_N\right]^T \notag
\end{align}
Since
\begin{align}
&\mathbb{E}\left({\mathbf{w}}(\iota){\mathbf{w}}^*(\iota)\right) = \mathbb{E}\left(\mathbf{m}_{\iota}^H\bar{\mathbf{w}}_{\iota} \bar{\mathbf{w}}_{\iota}^H \mathbf{m}_{\iota} \right)  = \sigma^2_{\bar{\mathbf{w}}},  \notag \\
&\mathbb{E}\left({\mathbf{w}}(\iota){\mathbf{w}}^*(\kappa)\right) = \mathbb{E}\left(\mathbf{m}_{\iota}^H\bar{\mathbf{w}}_{\iota} \bar{\mathbf{w}}_{\kappa}^H \mathbf{m}_{\kappa} \right) = 0,\; \forall \iota\neq \kappa \notag
\end{align}
the covariance of the equivalent noise $\mathbf{w}$ is thus $
 \mathbb{E}(\mathbf{w}\mathbf{w}^H) = \sigma^2_{\bar{\mathbf{w}}} \mathbf{I}_{N}$.
Let $\mathbf{n}\triangleq \boldsymbol{\nu} + \mathbf{w}$, as $\boldsymbol{\nu}$ and $\mathbf{w}$ are independent of each other, we have $\mathbf{n} \sim {\cal {CN}}\left(\mathbf{0}, \left( \sigma^2_{\bar{\mathbf{w}}} + \sigma^2_{\boldsymbol{\nu}}\right)\mathbf{I}_{N}\right)$.

Based on the above analysis, beam training for the link between BS/AP and MT is summarized as follows.\\
\vspace{0.1cm}
\noindent \textbf{\emph{Sub-problem 1: } } How to accurately estimate the parameter set $(\zeta_{LoS}, \delta_{BM,1}, \theta_{BM,1}, \phi_{BM,1})$ from $\mathbf{y} $.
\vspace{0.1cm}


\textbf{Step 2.} \underline{Activate the $i$-th IRS,  de-activate the  rest IRSs, and } \underline{estimate the parameters $(\delta_{BR_iM}, \theta_{R_iM}, \phi_{R_iM} \ominus \theta_{BR_i})$ of the} \underline{$i$-th VLoS path. Repeat the above process for the rest IRSs.}

As $\phi_{BR_i}$ is known, with the transmit beamforming vector $\mathbf{f} = \frac{\mathbf{a}_{B}(\phi_{BR_i})}{\sqrt{N_B}} $, BS/AP is able to concentrate its power towards IRSs via transmit beamforming. Simultaneously, IRS performs \emph{passive random reflection} and MT performs \emph{receive random beamforming}, the received signal at MT side is written as
\begin{align}
&\qquad  y\notag \\
& = \sqrt{P_{Tx}}\mathbf{m}^H \left( \mathbf{H}_{BM} +  \mathbf{H}_{BR_iM}  \right)  \frac{\mathbf{a}_{B}(\phi_{BR_i})}{\sqrt{N_B}}  + \mathbf{m}^H\bar{\mathbf{w}}\notag \\
& =   \sqrt{N_B P_{Tx}} \zeta_{VLoS,i}\delta_{B R_i M} \mathbf{m}^H \mathbf{a}_{M}(\theta_{R_iM}) \mathbf{a}_{R_i}^H(\phi_{R_iM}\ominus \theta_{BR_i}) \bar{\mathbf{g}}_{i}   \notag \\
& +  \underbrace{ \sqrt{P_{Tx}}\zeta_{LoS}\delta_{1} \mathbf{m}^H \mathbf{a}_{M}(\theta_{MB,1}) \mathbf{a}^H_{B}(\phi_{MB,1})\frac{\mathbf{a}_{B}(\phi_{BR_i})}{\sqrt{N_B}}}_{\nu_1}  \notag \\
&+\underbrace{\sum_{l=2}^{L}\sqrt{P_{Tx}}\delta_{l} \mathbf{m}^H \mathbf{a}_{M}(\theta_{MB,l}) \mathbf{a}^H_{B}(\phi_{MB,l})\frac{\mathbf{a}_{B}(\phi_{BR_i})}{\sqrt{N_B}}}_{\nu_2} + \underbrace{\mathbf{m}^H\bar{\mathbf{w}}}_{w}
\end{align}
The interference term $\nu_1$ and $\nu_2$ are insignificant due to 1) the small NLoS path coefficients $\delta_{l} (l = 2,\cdots, L) $ in mmWave band, 2) the spatial filtering impact, i.e., $\mathbf{a}^H_{B}(\theta_{MB, l})  \mathbf{a}_{B}(\phi_{BR_i}) \approx 0, (l = 1, 2,\cdots, L)$ for $|\phi_{BR_i} - \theta_{MB, l}| > \frac{1}{N_B}$.

Similar to \eqref{Sysmodel}, by concatenating $N$ channel measurements, we have
\begin{align}
 & \mathbf{y} = \sqrt{N_B P_{Tx}}\zeta_{VLoS,i} \delta_{B R_i M}  \mathbf{D} \mathbf{b} (\theta_{R_iM}, \phi_{R_iM} \ominus \theta_{BR_i}) \notag \\
 & \qquad + \underbrace{\boldsymbol{\nu}_1 + \boldsymbol{\nu}_2 + \mathbf{w}}_{\mathbf{n}}  \label{Sysmodel2}
\end{align}
where
\begin{align}
   &{\mathbf{D}}  = \left[ {\mathbf{g}}_1 \otimes \mathbf{m}_1^*,\;  {\mathbf{g}}_2  \otimes \mathbf{m}_2^*,\; \cdots, \; {\mathbf{g}}_N  \otimes \mathbf{m}_N^*\right]^T \notag
\end{align}

Based on the above analysis, beam training for the reflected path between BS/AP and MT assisted by the $i$-th IRS  is summarized as follows.\\
\vspace{0.1cm}
\noindent \textbf{\emph{Sub-problem 2: } }
How to accurately estimate  the parameter set $(\zeta_{VLoS,i}, \delta_{B R_i M}, \theta_{R_iM}, \phi_{R_iM}\ominus\theta_{BR_i})$ from   $\mathbf{y}$.
\vspace{0.1cm}

\begin{remark}
From \eqref{Sysmodel} and \eqref{Sysmodel2}, we can easily find that Sub-problem 1 and Sub-problem 2 are mathematically equivalent. Therefore, through flexible control over IRS, we are capable to decompose the complicated non-sparse channel estimation problem of IRSs assisted mmWave MIMO into a set of simplified sub-problems.
\end{remark}

{ 
\subsection{Protocol of Joint Beam Training and Positioning}

\begin{figure}[tp]{
\begin{center}{\includegraphics[ height=8.2cm]{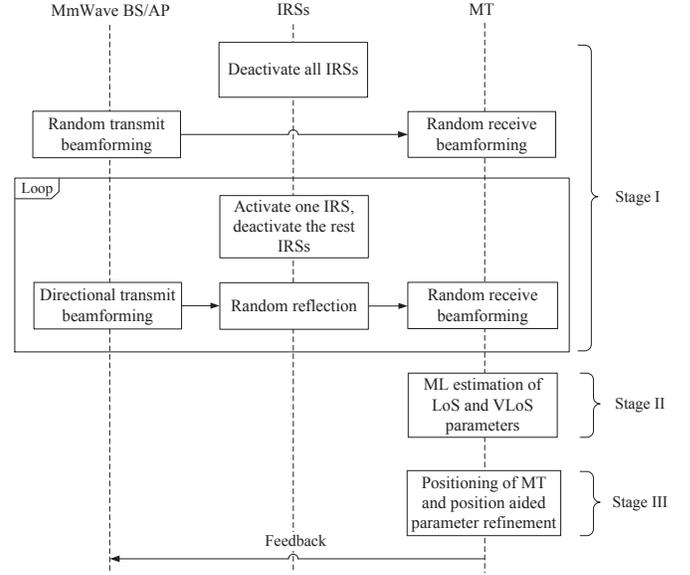}}
\caption{ Sequence diagram of joint beam training and positioning}\label{UML}
\end{center}}
\end{figure}
 }

On the basis of beam training breakdown, we introduce the protocol for joint beam training and position in the IRSs assisted mmWave communication system.

The procedures of the proposed scheme are given in \figref{UML}. Specifically, the scheme is divided into three  stages, i.e., Stage I. channel measurement, Stage II. parameter estimation and Stage III. positioning and location information aided parameter enhancement. In Stage I, the channel measurement vector $\mathbf{y}$ in Eq. \eqref{Sysmodel} and Eq. \eqref{Sysmodel2} are collected; In Stage II, ML estimation is performed to estimate the parameters of LoS path and VLoS paths in MT side, which will be introduced in Section IV; In Stage III, positioning and position aided path parameter refinement are performed in MT side, which will be introduced in Section V.

In practice, when an MT enters a cell, the prior information, e.g., random beamforming vector sequence of BS/AP, position of BS/AP and IRSs, will be sent to the MT via the lower frequency link, such as sub-6GHz link of 5G New Radio. Upon the request of high-speed mmWave data transmission, the random beamforming process in Stage I will be triggered periodically prior to data transmission to help setup initial beam alignment for new users and maintain beam alignment for the existing users. Then, in Stage II and Stage III,  each MT performs parameters estimation, positioning and position aided parameter refinement based on its own channel measurement vector $\mathbf{y}$. Finally, the estimated path parameters are fed back to BS/AP by each MT individually to facilitate beamforming designs for the subsequent mmWave data transmission.

It is noteworthy that, as random beamforming is quasi-omnidirectional \cite{myers2020deep}, the pilot sent by BS/AP can be received by MTs from all directions simultaneously. The broadcasting mechanism of random beamforming, which is similar to Global Positioning System (GPS), enables channel measurements to be collected and processed by each MT individually without causing interference. Therefore, training overhead of the proposed scheme will not increase with MT number, which renders the scheme particularly suitable for multi-user scenario.

\section{Beam Training With Random Beamforming -- Parameter Estimation and Feasibility Study }

In this section, ML estimation method is applied to estimate the path parameters $(\delta, \theta, \phi)$ of LoS/VLoS paths from channel measurements sampled by random Rx/Tx beamforming. Furthermore, the feasibility of random beamforming based beam training is verified.

\subsection{Maximum Log-likelihood Estimation of $(\delta, \theta, \phi)$}
For conciseness of expression, we write the unified model of sub-problem 1 and sub-problem 2 as
\begin{align} \label{UniModel}
 \mathbf{y} =  \zeta \delta \mathbf{D} \mathbf{b}(\theta,\phi) + {\mathbf{n}}
\end{align}
where  $\zeta$ is the indicator of blockage, $\delta$ is equivalent path gain ( $\delta = \sqrt{P_{Tx}}  \delta_{BM,1}$ or $\delta = \sqrt{P_{Tx} N_B} \delta_{B R_i M} $ ), $\theta$ is cosine AoA, $\phi$ is equivalent cosine AoD ($\phi = \phi_{BM,1}$ or $\phi = \phi_{R_iM}\ominus\theta_{BR_i}$), and
$\mathbf{b} (\theta, \phi) \triangleq  vec( \mathbf{a}_{Rx}(\theta) \mathbf{a}^H_{Tx}(\phi))$.

It is noteworthy that estimation of $(\delta, \theta, \phi)$ should be performed merely when $\zeta = 1$, as the measurement vector $\mathbf{y}$  given that $\zeta = 0$  contains no information about $(\delta, \theta, \phi)$. Therefore, we estimate the parameters $(\delta, \theta, \phi)$
through maximizing  log-likelihood function under the assumption that $\zeta = 1$, i.e.,
\begin{align} \label{MaxL}
(\hat{\delta}, \hat{\theta}, \hat{\phi}) = \argmax_{\delta, \theta, \phi} \mathcal{L}( \delta, \theta, \phi)
\end{align}
where
\begin{align} \label{LogLikeli}
\mathcal{L}( \delta, \theta, \phi)  = & \log P(\mathbf{y}|\zeta = 1, \delta, \theta, \phi) \notag \\
 =& -  N\log  {\pi} -  N\log \sigma^2 - \frac{\| \mathbf{y} -   \delta  \mathbf{D} \mathbf{b} (\theta, \phi) \|_2^2}{\sigma^2}
\end{align}
and the conditional probability is
\begin{align}
    P(\mathbf{y}| \zeta, \delta, \theta, \phi)  = \frac{1}{\pi^N \det(\sigma^2 \mathbf{I}_N)} e^{-\frac{(\mathbf{y} -  \zeta\delta  \mathbf{D} \mathbf{b} (\theta, \phi))^H  (\mathbf{y} -  \zeta\delta  \mathbf{D} \mathbf{b} (\theta, \phi))}{\sigma^{2}}   }
\end{align}

\subsubsection{Estimation of $\delta$}

Before the derivation of $\hat\theta, \hat\phi$, we should find the expression of $ \hat{\delta}$. To this end, we ignore terms independent thereof and set
\begin{align}
\frac{\partial \mathcal{L}(\delta, \theta, \phi)}{\partial \delta}  = 0 \label{Delta1}
\end{align}
Expanding Eq. \eqref{Delta1}, we have
\begin{align}
2Re\left\{\left(\mathbf{D}\mathbf{b}(\theta, \phi)\right)^H (\mathbf{y} -  \delta  \mathbf{D} \mathbf{b} (\theta, \phi))\right\} = 0 \label{Delta2}
\end{align}
From Eq. \eqref{Delta2},  the optimal $\hat{\delta}$ is derived as
\begin{align}
\hat{\delta}  =   \frac{\mathbf{b}^H(\theta, \phi)  \mathbf{D}^H  \mathbf{y}}{\| \mathbf{D} \mathbf{b}(\theta, \phi) \|_2^2} \label{OptDelta}
\end{align}

\subsubsection{Estimation of $\theta$ and $\phi$}
Next, we will jointly estimate  $\theta$ and $\phi$.  Substituting Eq. \eqref{OptDelta}  into Eq. \eqref{LogLikeli}, we have
\begin{align}
&\mathcal{L}(\delta, \theta, \phi)  \notag \\
 =& -  N\log  {\pi}  -   N\log \sigma^2 - \frac{ \left \| \mathbf{y}  - \frac{ \mathbf{D} \mathbf{b} (\theta, \phi) \mathbf{b}^H(\theta, \phi)  \mathbf{D}^H}{\| \mathbf{D} \mathbf{b}(\theta, \phi) \|_2^2} \mathbf{y} \right \|_2^2}{\sigma^2}
\end{align}
Since
\begin{align}
& \left \| \mathbf{y}  - \frac{ \mathbf{D} \mathbf{b} (\theta, \phi) \mathbf{b}^H(\theta, \phi)  \mathbf{D}^H}{\| \mathbf{D} \mathbf{b}(\theta, \phi) \|_2^2} \mathbf{y} \right \|_2^2 \notag \\
=& \mathbf{y}^H (\mathbf{I}  - \frac{ \mathbf{D} \mathbf{b} (\theta, \phi) \mathbf{b}^H(\theta, \phi)  \mathbf{D}^H}{\| \mathbf{D} \mathbf{b}(\theta, \phi) \|_2^2})  \mathbf{y},
\end{align}
the beam training problem is formulated as
\begin{align}
P1: \quad & \max_{\theta, \phi}  \left\| \frac{\mathbf{b}^H (\theta, \phi)\mathbf{D}^H}{\|  \mathbf{D} \mathbf{b} (\theta, \phi)\|_2}\mathbf{y} \right\|_2^2 \notag \\
&s.t. \; -1\leq \theta < 1 \notag\\
& \;\;\;\;\;\; -1\leq \phi <1 \notag
\end{align}

P1 is a non-convex problem. However, as there are only two real-valued variables to be estimated, a simple yet efficient two-step algorithm can be readily applied to solve P1.
For conciseness, let $g({\theta, \phi}) \triangleq  \left\| \frac{\mathbf{b}^H (\theta, \phi)\mathbf{D}^H}{\|  \mathbf{D} \mathbf{b} (\theta, \phi)\|_2}\mathbf{y} \right\|_2^2$. The two-step algorithm is explained as follows.

\vspace{0.02cm}

\noindent \textbf{Step 1. Joint AoA and AoD Coarse Search}

Set quantization level $Z_{\theta}$ and $Z_{\phi}$, and then exhaustively search for  the $N_{pk}$ largest maxima that satisfy
\begin{align}
g(\theta_{\hat\iota},  \phi_{\hat\kappa}) > g( \theta_{\hat\iota-1},  \phi_{\hat\kappa})\notag \\
g( \theta_{\hat\iota},  \phi_{\hat\kappa}) > g( \theta_{\hat\iota+1},  \phi_{\hat\kappa}) \notag\\
g( \theta_{\hat\iota},  \phi_{\hat\kappa}) > g( \theta_{\hat\iota},  \phi_{\hat\kappa-1}) \notag \\
g( \theta_{\hat\iota},  \phi_{\hat\kappa}) > g( \theta_{\hat\iota},  \phi_{\hat\kappa+1}) \notag
\end{align}
over the discrete grid
\begin{align}
&\mathcal{D} \triangleq \left\{(\theta_{\iota}, \phi_{\kappa}) \Big| \;\theta_{\iota} = -1 + \frac{2\iota-1}{Z_{\theta}}, \iota = 1, 2, \cdots, Z_{\theta}, \right.\notag \\
&\left. \qquad \qquad \quad \;\; \phi_{\kappa} = -1 + \frac{2\kappa-1}{Z_{\phi}}, \kappa = 1, 2, \cdots, Z_{\phi} \right\}
\end{align}

\noindent \textbf{Step 2. Joint AoA and AoD Fine Search}

For a given discrete maximum $(\theta_{\hat \iota}, \phi_{\hat \kappa})^T$, run gradient descent search starting from $( \theta^{(1)},  \phi^{(1)})^T = (\theta_{\hat \iota}, \phi_{\hat \kappa})^T $ as follows
\begin{align}
 \left(
   \begin{array}{c}
     \theta^{(i+1)} \\
     \phi^{(i+1)} \\
   \end{array}
 \right)  =    \left(
   \begin{array}{c}
     \theta^{(i)} \\
     \phi^{(i)} \\
   \end{array}
 \right)    \oplus \lambda    \left(
   \begin{array}{c}
     \frac{\partial g({\theta, \phi})}{\partial \theta} \big|_{\theta = \theta^{(i)}} \\
     \frac{\partial g({\theta, \phi})}{\partial \phi} \big|_{\phi = \phi^{(i)}} \\
   \end{array}
 \right)
\end{align}
{  where $\lambda$ is the preset step size and the expressions of $ \frac{\partial g({\theta, \phi})}{\partial \theta}$ and $\frac{\partial g({\theta, \phi})}{\partial \phi} $ are given in Appendix A.}
The iteration stops when $(\theta^{(i+1)} \ominus \theta^{(i)} )^2 + (\phi^{(i+1)} \ominus \phi^{(i)} )^2 \leq \epsilon$, where $\epsilon$ is a preset parameter.

Repeat the above operations over the rest $N_{pk}-1$ maxima derived in Step 1, and select the best one as  $(\hat{\theta}, \hat{\phi})$. Then,  the exact value of the estimated path gain $\hat{\delta}$ can be subsequently obtained by substituting $(\hat{\theta}, \hat{\phi})$ into Eq. \eqref{OptDelta}.

{ 
\begin{remark}
The  complexity of Step 1 is $\mathcal{O}(2^2 Z_{\phi} Z_{\theta})$. The  complexity of Step 2 mainly arises from the computation of the gradients  $ \frac{\partial g({\theta, \phi})}{\partial \theta}$ and $\frac{\partial g({\theta, \phi})}{\partial \phi} $, which, according to Eq. \eqref{Gradient1} and Eq. \eqref{Gradient2},  is $\mathcal{O}(N_BN_MN)$ (or $\mathcal{O}(N_{R_i}N_MN)$). Hence, the complexity of Step 2 is $\mathcal{O}(n_{iter}N_{pk}N_BN_MN)$ or ($\mathcal{O}(n_{iter}N_{pk}N_{R_i}N_MN)$), where the iteration number $n_{iter}$ depends on step size and stopping criterion of the gradient method and is generally less than $20$. Thus, the overall complexity is $\mathcal{O}(2^2 Z_{\phi} Z_{\theta} + n_{iter}N_{pk}N_{B}N_MN)$ (or $\mathcal{O}(2^2 Z_{\phi} Z_{\theta} + n_{iter}N_{pk}N_{R_i}N_MN)$).

\end{remark}
}



\subsection{ Uniqueness of The Estimated AoA and AoD Pair}
To delve into the effectiveness of beam training with random beamforming, conditions under which $(\theta, \phi)$ can be accurately estimated from the measurement signal $\mathbf{y}$ are studied  in the ideal scenario without noise or interference.

Firstly, two definitions of uniqueness are introduced as follows.
 \\
(1) \textbf{Uniqueness of measurement signal representation}, namely
\begin{align} \label{UniRresentation}
 \mathbf{y} &= \delta \mathbf{D} \mathbf{b} (\theta, \phi) \notag \\
 & \neq  \widetilde{\delta} \mathbf{D} \mathbf{b} (\widetilde{\theta}, \widetilde{\phi}),\;\; \;\; \forall \widetilde{\delta} \in \mathbb{C},  \forall (\widetilde{\theta}, \widetilde{\phi})  \neq (\theta, \phi)
\end{align}
\noindent (2) \textbf{Uniqueness of estimated AoA and AoD pair}, namely
\begin{align} \label{UniEst}
\left\| \frac{\mathbf{b}^H ( {\theta},  {\phi})\mathbf{D}^H}{\|  \mathbf{D} \mathbf{b} ( {\theta},  {\phi})\|_2} \mathbf{y} \right\|_2  >   \left\| \frac{\mathbf{b}^H (\widetilde{\theta}, \widetilde{\phi})\mathbf{D}^H}{\|  \mathbf{D} \mathbf{b} (\widetilde{\theta}, \widetilde{\phi})\|_2}  \mathbf{y}\right\|_2, \; \forall (\widetilde{\theta}, \widetilde{\phi})  \neq (\theta, \phi)
\end{align}
Uniqueness of measurement signal representation means that any AoA, AoD pair $(\widetilde\theta, \widetilde\phi)$ that differs from $(\theta, \phi)$ cannot construct the measurement signal $\mathbf{y}$. It is an inherent property of the \emph{sampling method},  which is primarily determined by $\mathbf{D}$.   By contrast, uniqueness of the estimated AoA and AoD depends on both  \emph{sampling method} and \emph{estimation method}. It indicates that AoA, AoD pair can be accurately estimated from the measurement signal $\mathbf{y}$ using a specific estimation method. 

In the following Theorem, we will study the relationship between the above two types of uniqueness.
\begin{theorem}
As long as uniqueness of measurement signal representation is satisfied, ML method is capable to accurately estimate the AoA, AoD pair.
\end{theorem}
\begin{proof}
See Appendix B.
\end{proof}



According to Theorem 1, the uniqueness of AoA and AoD estimation is equivalent to the uniqueness of measurement signal representation, which means we just need to investigate the conditions on which uniqueness of measurement signal representation can be achieved.

Before studying the sensing matrix $\mathbf{D}$, we will observe the signal space of channel response.
The vectorized response of LoS path, namely $\mathbf{h} = \delta \mathbf{b}(\theta, \phi)$, is a high dimensional ($N_rN_t$-dimensional) variable that is characterized by $(\delta, \theta, \phi)$, and we define the signal space of $\mathbf{h}$ as
\begin{align}
\mathcal{S} \triangleq \{ \delta \mathbf{b}(\theta, \phi)| \delta\in \mathbb{C}, -1\leq \theta, \phi <1 \}
\end{align}
$\mathcal{S}$ is a nonlinear $k$-dimensional ($k=3$) submanifold of $\mathbb{C}^{N_rN_t}$ with the parameters $(\delta, \theta, \phi)$ \cite{baraniuk2009random,clarkson2008tighter}. As $\mathbf{b}(\theta, \phi)$ is the Kronecker product of two array steering vectors, $\mathcal{S}$ is indeed the so-called \emph{array manifold} \cite{Manifold1}. Thus, one channel realization $\check{\mathbf{h}}$ with the parameters $(\check{\delta}, \check{\theta}, \check{\phi})$ can  be seen as a point in the array manifold. The dimensionality $k$ can be interpreted as an ``information level" of the signal, analogous to the sparsity level in compressive sensing problems \cite{baraniuk2009random,CompressedSensingBook,candes2006stable}. In \cite{baraniuk2009random},  it is proved that signals obeying manifold models can also be recovered from only a few measurements, simply by replacing the traditional compressive sensing model of sparsity with a manifold model for $\mathbf{h}$.  The above statement is supported by Lemma 1.

\begin{lemma}
For a random orthoprojector $\boldsymbol{\Phi}\in \mathbb{C}^{M \times N}$, the following statement
\begin{align}
(1-\epsilon)\sqrt{\frac{M}{N}}\leq \frac{\left\| \boldsymbol{\Phi}\mathbf{h}_{1} - \boldsymbol{\Phi}\mathbf{h}_{2} \right\|_2^2}{\|  \mathbf{h}_1 - \mathbf{h}_2 \|_2^2} \leq (1+\epsilon)\sqrt{\frac{M}{N}},\notag\\
   \forall \mathbf{h}_1, \mathbf{h}_2 \in \mathcal{S},  \mathbf{h}_1 \neq  \mathbf{h}_2
\end{align}
holds with high probability, when dimensionality $M$ of the projected low-dimensional space is sufficient \footnotemark, where $\mathbf{h}_1\in \mathcal{S}, \mathbf{h}_2 \in \mathcal{S}$, $\mathbf{h}_1 \neq \mathbf{h}_2$, $0<\epsilon<1$ is the isometry constant \cite{baraniuk2009random}.
\end{lemma}

\footnotetext{The sufficient number of $M$ is related to $\epsilon$ and several manifold-related factors, e.g.,  condition number, volume, and geodesic covering regularity. Detailed analysis can be referred to \cite{baraniuk2009random,clarkson2008tighter}. In practice, the exact relationship between the sufficient number and its dependent factors is of limited significance due to the following two reasons, (1) the received measurement signal is corrupted by noise, (2)$M$ can be online adjusted according to channel conditions.}

\begin{remark}
$\|\mathbf{h}_1 - \mathbf{h}_2 \|_2^2$ is the Euclidean distance between two points $\mathbf{h}_1 $, $\mathbf{h}_2$ on the manifold, and $\left\| \boldsymbol{\Phi}\mathbf{h}_{1} - \boldsymbol{\Phi}\mathbf{h}_{2} \right\|_2^2 $ is the Euclidean distance between the projected points $\boldsymbol{\Phi}\mathbf{h}_{1}, \boldsymbol{\Phi}\mathbf{h}_{2}$ on the image of $\mathcal{S}$ (namely $\boldsymbol{\Phi}\mathcal{S}$). The isometry constant $\epsilon$  measures the degree that the pairwise Euclidean distance between points on $\mathcal{S}$ is preserved under the mapping $\boldsymbol{\Phi}$. Apparently, Lemma 1 indicates that $\left\| \boldsymbol{\Phi}\mathbf{h}_{1} - \boldsymbol{\Phi}\mathbf{h}_{2} \right\|_2^2>0$ is satisfied with high probability, as it is a weaker condition than Lemma 1.
\end{remark}

Although the sensing matrix $\mathbf{D}$ is not necessarily an orthoprojector,  via singular value decomposition, it can be decomposed as $\mathbf{D} = \widetilde{\boldsymbol{\Psi}}\widetilde{\boldsymbol{\Lambda}} \widetilde{\boldsymbol{\Phi}}$, where $\widetilde{\boldsymbol{\Psi}} \in \mathbb{C}^{M\times M},  \widetilde{\boldsymbol{\Lambda}} \in \mathbb{C}^{M\times M}$, and $\widetilde{\boldsymbol{\Phi}} \in  \mathbb{C}^{M\times N}$. Then, we have $\|\mathbf{D}\mathbf{h}_{1} - \mathbf{D}\mathbf{h}_{2} \|_2^2  =  \| \widetilde{\boldsymbol{\Lambda}} \widetilde{\boldsymbol{\Phi}}\mathbf{h}_{1} - \widetilde{\boldsymbol{\Lambda}}\widetilde{\boldsymbol{\Phi}}\mathbf{h}_{2} \|_2^2$, where
$\widetilde{\boldsymbol{\Phi}}$ is indeed the orthoprojector, and $\widetilde{\boldsymbol{\Lambda}} $ is a diagonal matrix with non-zero elements that scales the component in each dimension.    $\|\widetilde{\boldsymbol{\Phi}}\mathbf{h}_{1} - \widetilde{\boldsymbol{\Phi}}\mathbf{h}_{2} \|_2^2> 0$ implicates $ \|\mathbf{D}\mathbf{h}_{1} - \mathbf{D}\mathbf{h}_{2} \|_2^2> 0$, which is equivalent to $\mathbf{D}\mathbf{h}_{1} \neq \mathbf{D}\mathbf{h}_{2}$, namely, $
\delta_1\mathbf{D}\mathbf{b}(\theta_1,\phi_1) \neq  \delta_2\mathbf{D}\mathbf{b}(\theta_2, \phi_2),\; \forall (\delta_1, \theta_1, \phi_1) \neq (\delta_2, \theta_2, \phi_2)$.
Thus, it is easy to find that $ \mathbf{D}\mathbf{b}(\theta_1,\phi_1) \neq  {\mu} \mathbf{D}\mathbf{b}(\theta_2, \phi_2),\; \forall (  \theta_1, \phi_1) \neq ( \theta_2, \phi_2), \forall  {\mu}\in \mathbb{C}$,
where $ {\mu} \triangleq \frac{\delta_2}{\delta_1}$.

To conclude, the randomly generated sensing matrix $\mathbf{D}$ has a large probability to guarantee the uniqueness of ML based joint AoA and AoD estimation.

\subsection{On The Impact of Training Length $N$}
Theorem 1 indicates that, with random beamforming,  Eq. \eqref{UniEst} holds with high probability. In other words, in noiseless scenario, the distance gap between the highest peak  (global optimum) and other peaks (other local optimums) exist with high probability. However, in practice, corrupted by noise and interference, the highest peak may (1) shift to its adjacent points, or (2) be transcended and replaced by other peaks. Error Type 1 incurs mild AoA, AoD estimation error followed by power loss of an acceptable level; Error Type 2 incurs significant AoA, AoD estimation error followed by beam misalignment. Apparently, we would like to avoid Error Type 2.

To study the estimation error, the pairwise error probability (PEP) of any two parameter sets $(\theta, \phi)$ and $(\widetilde{\theta}, \widetilde{\phi})$ is derived in the following theorem.
\begin{theorem}
The PEP $Pe\left((\theta, \phi) \rightarrow (\widetilde{\theta}, \widetilde{\phi})\right)$ that $(\theta, \phi) $  is mistaken as $(\widetilde{\theta}, \widetilde{\phi})$ in relatively high SNR regime can be approximated as
\begin{align}\label{Th1}
 \qquad Pe\left((\theta, \phi) \rightarrow (\widetilde{\theta}, \widetilde{\phi})\right) \approx   Q\left(\frac{|\delta|^2}{2\sigma^2}  d^2(\mathbf{D}, \theta, \phi, \widetilde{\theta}, \widetilde{\phi}) \right)
\end{align}
{ where the Q-function is the tail distribution function of the standard normal distribution \cite{chiani2003new}}, and
\begin{align}
 d^2(\mathbf{D}, \theta, \phi, \widetilde{\theta}, \widetilde{\phi})  \triangleq \left\| \mathbf{D} \mathbf{b} (\theta, \phi)\right\|_2^2  -  \frac{|\mathbf{b}^H (\widetilde{\theta}, \widetilde{\phi})\mathbf{D}^H \mathbf{D} \mathbf{b} (\theta, \phi)|^2 }{\|  \mathbf{D} \mathbf{b} (\widetilde{\theta}, \widetilde{\phi})\|_2} \notag
\end{align}
\end{theorem}
\begin{proof}
See Appendix C.
\end{proof}

\begin{figure}[t]
\begin{minipage}[!h]{0.48\linewidth}
\centering
\includegraphics[ width=1.1\textwidth]{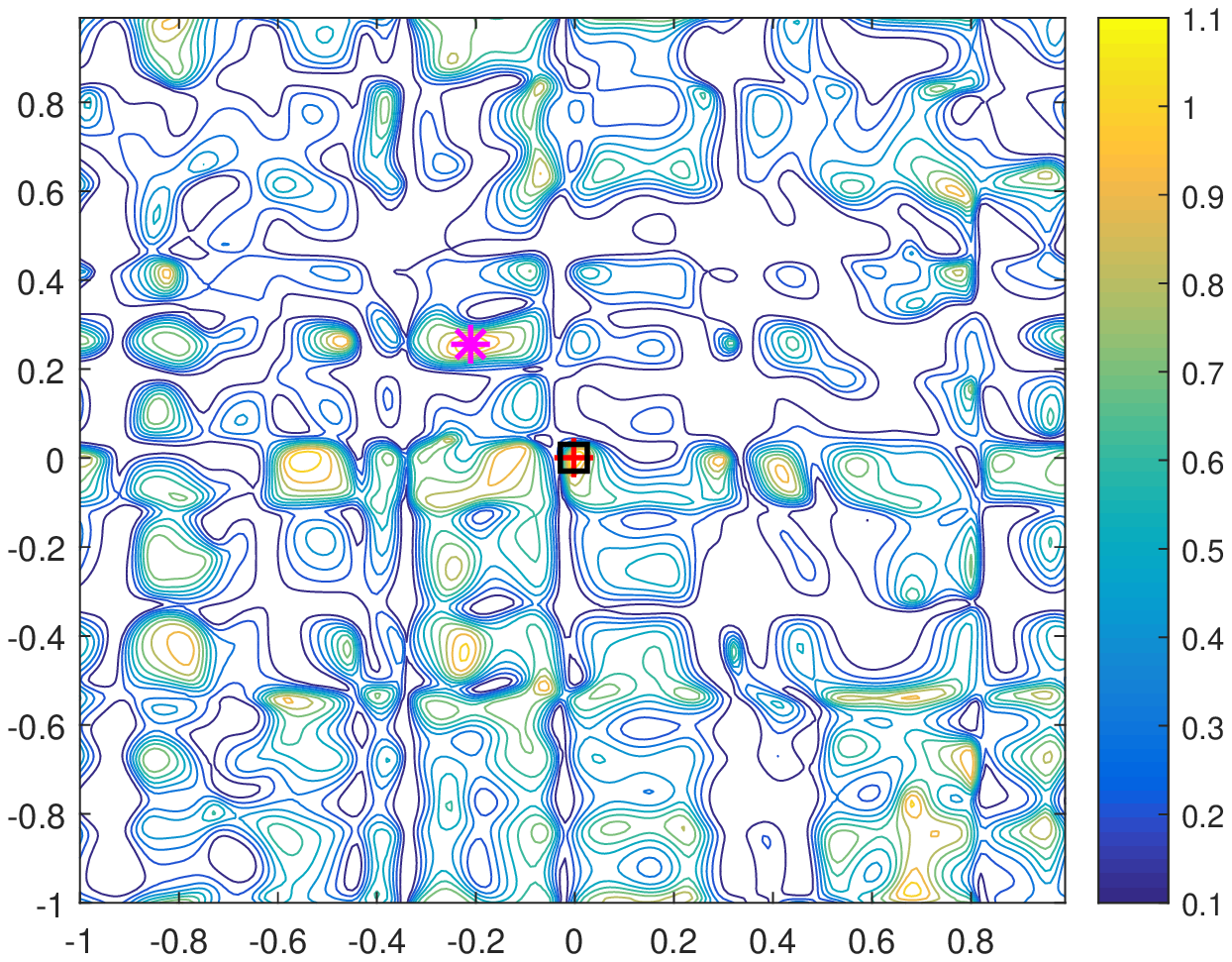}
\subcaption{$N=4$ }
\label{N4}
\end{minipage}
\begin{minipage}[!h]{0.48\linewidth}
\centering
\includegraphics[ width=1.1\textwidth]{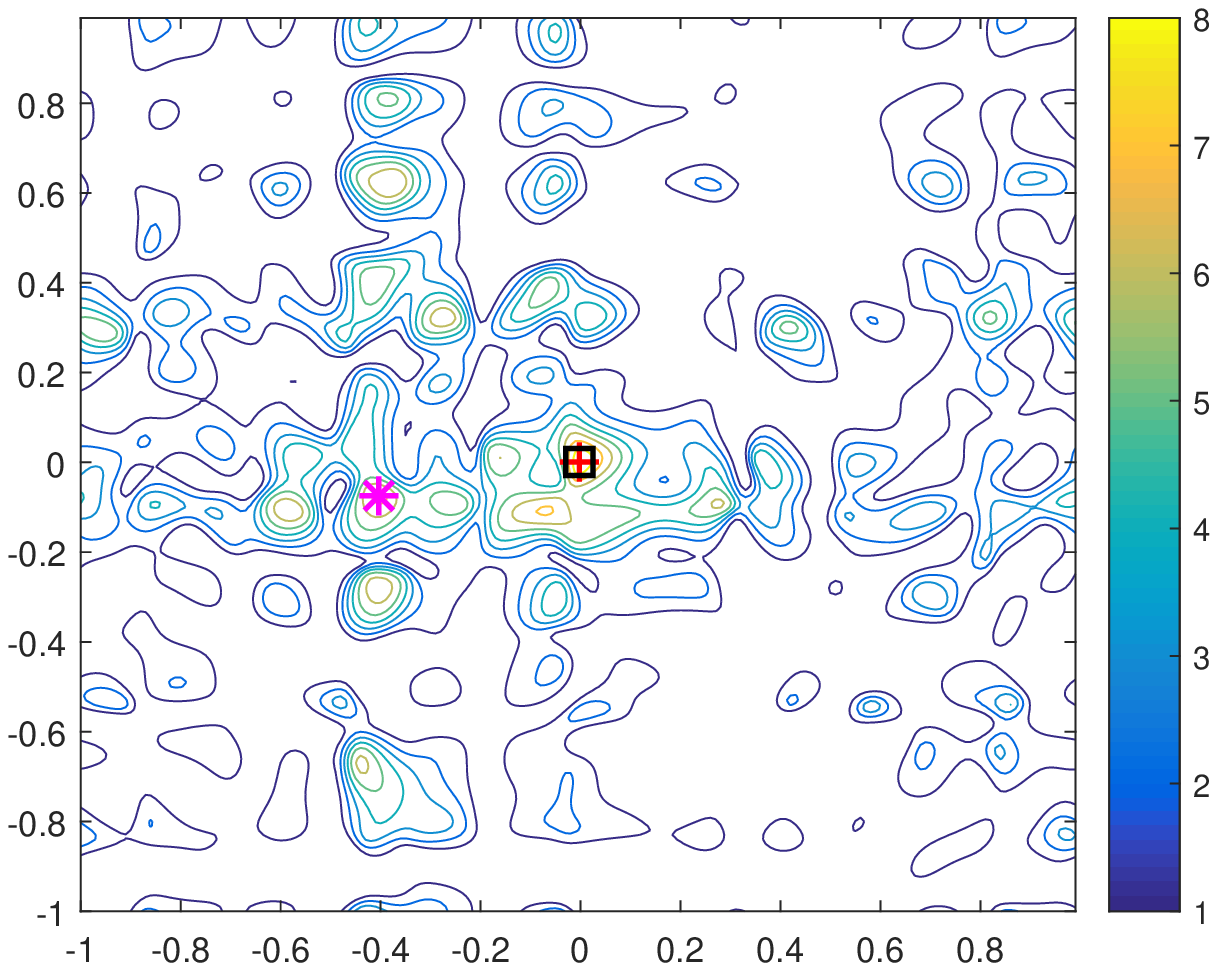}
\subcaption{$N=8$}
\label{N8}
\end{minipage}
\begin{minipage}[!h]{0.48\linewidth}
\centering
\includegraphics[ width=1.1\textwidth]{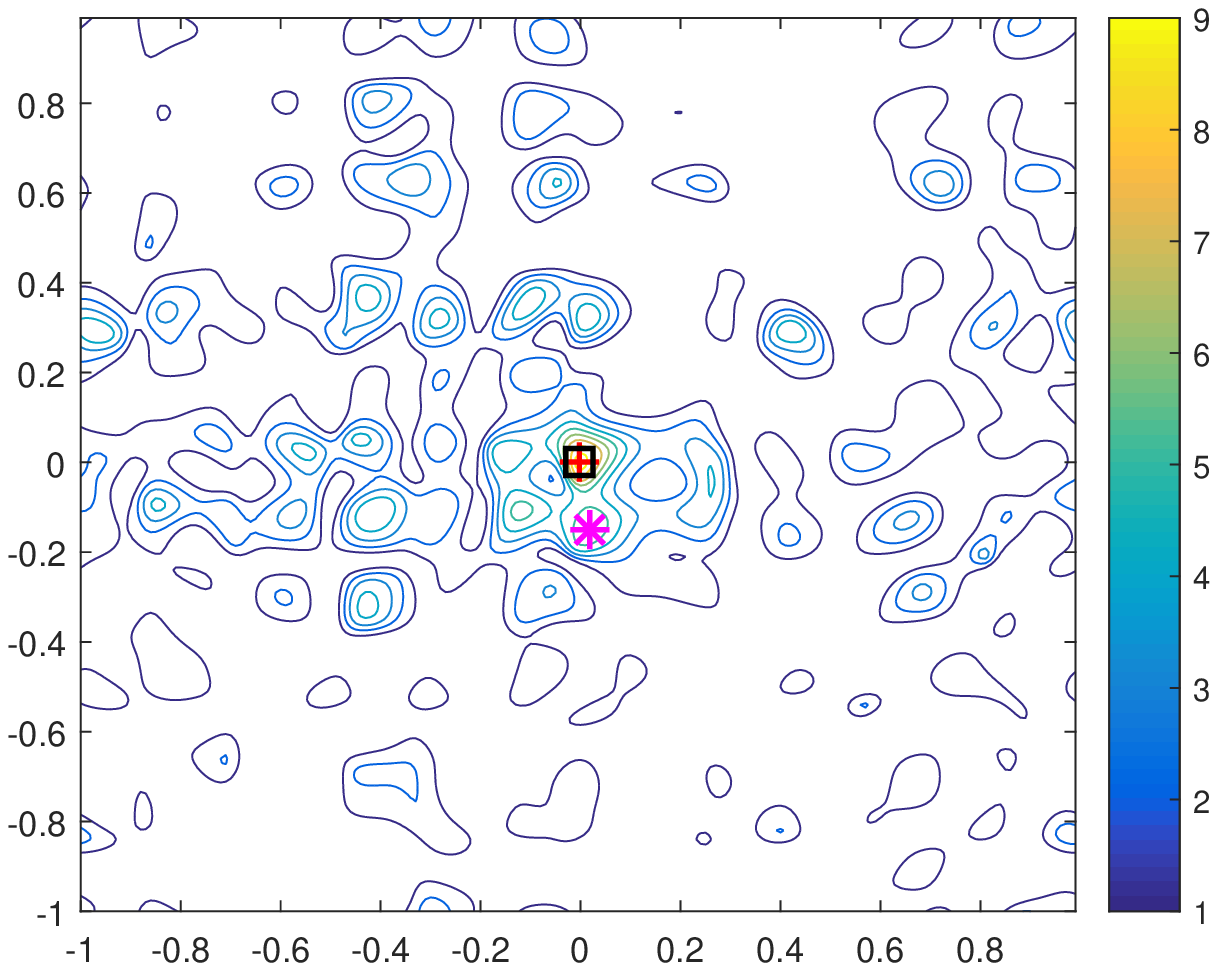}
\subcaption{$N=12$ }
\label{N12}
\end{minipage}
\hspace{0.17cm}\begin{minipage}[!h]{0.48\linewidth}
\centering
\includegraphics[ width=1.1\textwidth]{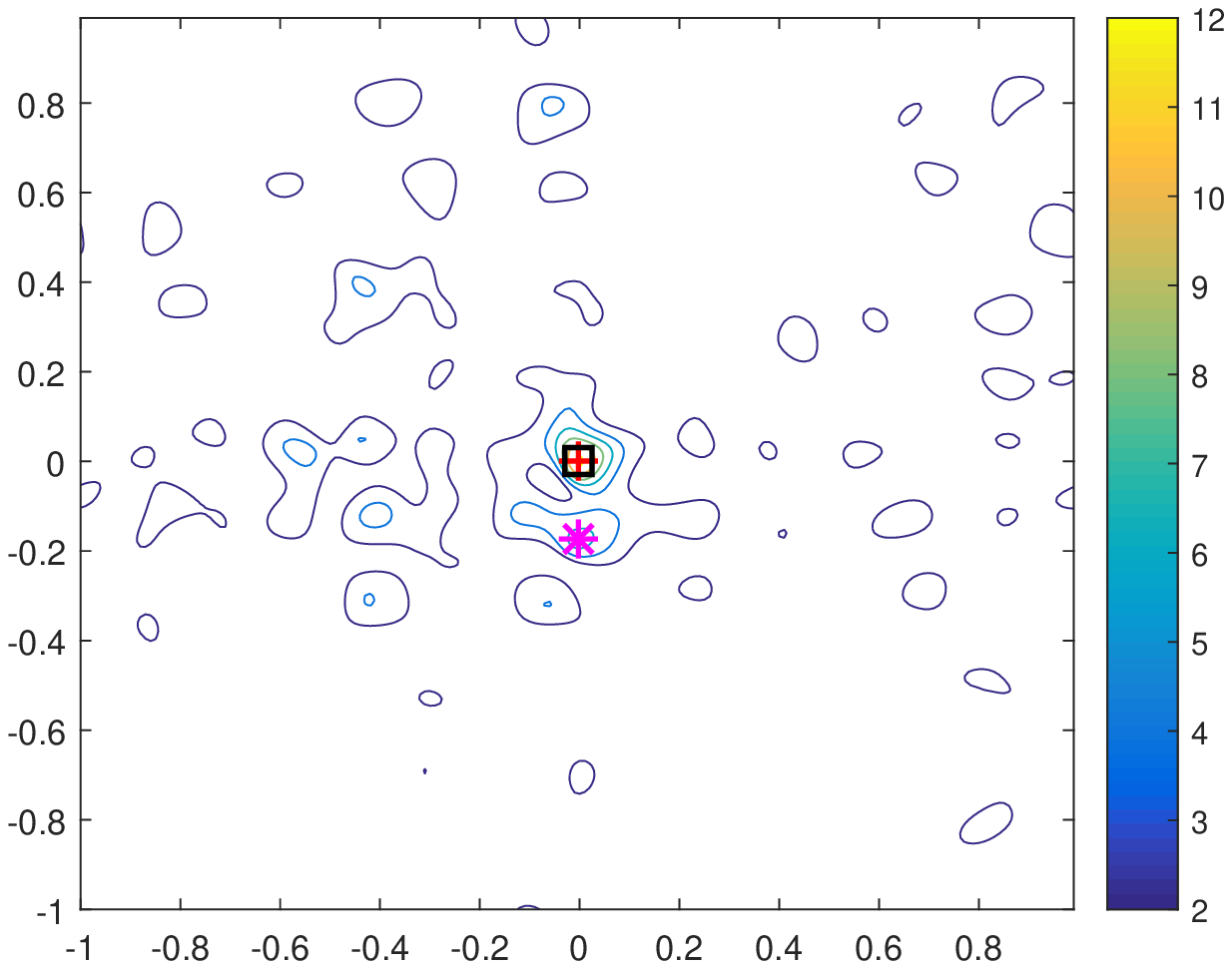}
\subcaption{$N=16$}
\label{N16}
\end{minipage}
\caption{Contour plots of $g(\theta,\phi)$ with different training lengths (Red cross represents the position of the first peak, purple asterisk represents the position of the second peak, and black square is the position of the actual AoA, AoD pair)}
\label{Contour}
\end{figure}

\begin{table}[tp] \centering
\noindent
\caption{Peak values of $g(\theta,\phi)$ over training length}\label{T1}
\begin{tabular}{|l|*{5}{c|}}\hline
\makebox[7em] {Training Length}
&\makebox[5em]{Peak 1}&\makebox[5em]{Peak 2} &\makebox[7em]{Peak 1 $-$ Peak 2}\\\hline\hline
\qquad $\;N=4$  & 1.1156   &  1.1044 &  0.0112    \\\hline
\qquad $\;N=8$   & 8.7223   &  7.2658  &  1.4573   \\\hline
\qquad $N=12$    & 9.4986   &  5.8000 &  3.6986    \\\hline
\qquad $N=16$    & 12.3338  &  6.6508  &  5.6830   \\\hline
\end{tabular}

\end{table}

Theorem 2 indicates that PEP is inversely proportional to $d^2(\mathbf{D}, \theta, \phi, \widetilde{\theta}, \widetilde{\phi})$. To build the connection between PEP and training length $N$, Proposition 1 is derived.

\begin{proposition}
$d^2(\mathbf{D}_N, \theta, \phi, \widetilde{\theta}, \widetilde{\phi})$  is monotonically  increasing over training length $N$, where $\mathbf{D}_{N} = \left[\mathbf{D}_{N-1}^H \; \mathbf{d}_{N}\right]^H$,  i.e.,
\begin{align}
d^2(\mathbf{D}_N, \theta, \phi, \widetilde{\theta}, \widetilde{\phi}) \geq  d^2(\mathbf{D}_{N-1}, \theta, \phi, \widetilde{\theta}, \widetilde{\phi})
\end{align}
and the equality holds only if
\begin{align}
&\frac{\mathbf{b}^H(\widetilde{\theta},\widetilde{\phi})\mathbf{d}_N\mathbf{d}_N^H  \mathbf{b}(\widetilde{\theta},\widetilde{\phi})}{ \mathbf{b}^H(\widetilde{\theta},\widetilde{\phi})\mathbf{d}_N\mathbf{d}_N^H  \mathbf{b}( {\theta}, {\phi}) }  =
&\frac{\mathbf{b}^H(\widetilde{\theta},\widetilde{\phi})\mathbf{D}_{N-1}^H  \mathbf{D}_{N-1} \mathbf{b}(\widetilde{\theta},\widetilde{\phi})}{\mathbf{b}^H(\widetilde{\theta},\widetilde{\phi})\mathbf{D}_{N-1}^H  \mathbf{D}_{N-1} \mathbf{b}(\theta,\phi)  }\notag \\
\end{align}
\end{proposition}
\begin{proof}
See Appendix C.
\end{proof}

To verify Proposition 1, we plot the contour of $g(\theta,\phi)$ with different training lengths  in noiseless scenario in \figref{Contour}.  We set $\delta=1, \theta =0, \phi = 0$.  As can be seen that the gap between the first and the second peaks increases over training length, and the value of which is given in Table I. In addition, we can find that position of the first peak is invariant to training length and remains the same as the actual AoA, AoD pair, while position of the second peak varies. This verifies the uniqueness of ML based joint AoA, AoD estimation.

\begin{remark}
According to Proposition 1, with random beamforming, the PEP probability of an erroneous estimate $(\widetilde{\theta},\widetilde{\phi})$ being mistaken as the authentic parameters $(\theta,\phi)$ decreases almost surely over  training length $N$. Therefore, an appropriate $N$ can guarantee a satisfying accuracy of parameter estimation in scenarios with different SNR and interference levels.
\end{remark}

\section{Interplay Between Positioning and Beam Training}

In IRSs assisted mmWave MIMO system, BS/AP and IRSs, with their positions and array directions being known by all the MTs, can be seen as anchor nodes or
beacons. The AoDs derived at beam training stage enable  MT to estimate its own position. Hence, IRSs assisted mmWave MIMO system is endowed with the capability of high-accuracy localization. The acquired position information is not only a fringe benefit, but also in turn facilitates beam training. The interplay between beam training and indoor positioning is explained as follows.   AoD estimate of the unblocked reliable links can yield the position of MT, and the position of MT, associated with anchor positions and anchor directions, can improve the precision of AoD/AoA estimation and assist in the decision of blockage indicator $\zeta$.


\subsection{Reliability of The Estimated AoA, AoD Pair $(\hat{\theta}, \hat{\phi})$}
To be concise, we treat BS/AP and IRSs as identical anchor nodes. The $\eta = 1$-st anchor is  BS/AP and the rest $N_{IRS}$ anchors  ($\eta = 2, 3, \cdots, N_{IRS}+1$) are IRSs. Although we have already obtained $N_{IRS} + 1$ sets of path parameters $(\hat{\delta}_\eta, \hat{\theta}_\eta, \hat{\phi}_\eta)$, we should be aware that the estimation is performed under the assumption that $\zeta_\eta = 1$. In practice, LoS and VLoS paths may suffer from blockage (namely $\zeta_\eta = 0$) by moving obstacles, which will jeopardize the estimation of $(\hat{\delta}_\eta, \hat{\theta}_\eta, \hat{\phi}_\eta)$.  Other than blockage, insufficient training length or low SNR may incur Error Type 2  of joint AoA and AoD estimation, which is defined in Section IV. C.

Therefore, it is essential to select the trustworthy parameters as the input of  positioning algorithm. To this end, we introduce the metric -- residual signal power ratio $\varpi_{\eta}$,  to measure the reliability of $(\hat{\delta}_\eta, \hat{\theta}_\eta, \hat{\phi}_\eta)$, i.e.,
\begin{align}
\varpi_{\eta} = \frac{\|\mathbf{y}_\eta - \hat{\delta}_\eta \mathbf{D} \mathbf{b}(\hat{\theta}_\eta, \hat{\phi}_\eta)  \|_2^2}{\| \mathbf{y}_\eta\|_2^2}
\end{align}
Recall that  $(\hat\delta_\eta, \hat\theta_\eta, \hat\phi_\eta)$   are obtained by minimizing $\|\mathbf{y}_{\eta} - \delta_\eta \mathbf{D} \mathbf{b}(\theta_\eta, \phi_\eta) \|_2^2$, the yielded estimate $(\hat{\delta}_\eta, \hat{\theta}_\eta, \hat{\phi}_\eta)$ will thus always result in  $\|\mathbf{y}_{\eta} - \delta_\eta \mathbf{D} \mathbf{b}(\theta_\eta, \phi_\eta) \|_2^2 \leq \|\mathbf{y}_{\eta}\|_2^2$. Therefore,  the range of $\varpi_{\eta}$ is $\varpi_{\eta} \in [0, 1]$.

Since the dominant component of mmWave channel is LoS path, the reconstructed signal $\hat{\delta}_\eta \mathbf{D} \mathbf{b}(\hat{\theta}_\eta, \hat{\phi}_\eta)$  should account for the majority of the received signal $\mathbf{y}$ given that the parameters $(\hat{\delta}_\eta, \hat{\theta}_\eta, \hat{\phi}_\eta)$ are accurate and residual signal power ratio $\varpi_{\eta}$ should be smaller. Conversely, when blockage or Error Type 2 occurs, the parameters $(\hat\delta_\eta, \hat\theta_\eta, \hat\phi_\eta)$  are heavily biased, and thus $\varpi_{\eta}$ should be larger. Following the above heuristics, anchors' reliability can be sorted.

\subsection{AoD Based Positioning}





\subsubsection{Geometric Relationship Between AoDs and MT Position}
We denote the index set of the reliable links as $\mathcal{N}$, position coordinates of the $\eta$-th anchor  as $\mathbf{p}_\eta$,  ULA direction of the $\eta$-th anchor as  $\mathbf{e}_\eta$. Note that $\mathbf{p}_\eta,   \mathbf{e}_\eta$ are known by MTs. The direction vector of the LoS path between MT and the $\eta$-th anchor is $\frac{\mathbf{p} - \mathbf{p}_{\eta}}{\| \mathbf{p} - \mathbf{p}_\eta\|_2}$, where $\mathbf{p}$ is the position of MT. Thus, the geometric relationship between AoDs and MT position is expressed as
\begin{align}
 \hat{\phi}_\eta = \underbrace{\frac{(\mathbf{p} - \mathbf{p}_{\eta})^T \mathbf{e}_\eta}{\| \mathbf{p} - \mathbf{p}_\eta\|_2}}_{\phi_\eta(\mathbf{p})} + \varepsilon_\eta, \;\; \eta \in \mathcal{N} \label{MeaAoD}
\end{align}
where $ \hat{\phi}_\eta$ is the estimate of cosine AoD of the $\eta$-th   link derived in beam training stage, $ \phi_\eta(\mathbf{p})$ is the actual cosine AoD that is dependent on position $\mathbf{p}$, and $\varepsilon_\eta$ is estimation error. For illustrative purposes, a typical scenario of IRSs assisted mmWave communications is shown in \figref{ScenarioPositioning1}.

\subsubsection{Taylor Series Method for AoD Based Positioning}
In the ideal case, when $\varepsilon_\eta = 0$, we have $\hat{\phi}_\eta =  \phi_\eta(\mathbf{p})$. The equation $\phi_\eta(\mathbf{p}) =  \frac{(\mathbf{p} - \mathbf{p}_{\eta})^T \mathbf{e}_\eta}{\| \mathbf{p} - \mathbf{p}_\eta\|_2} $ corresponds to a right circular cone. There are $3$ unknown variables of MT's position coordinates, thus the minimum sufficient number of unblocked links to estimate the 3-D position of MT is $|\mathcal{N}|= 3$, which is the intersection of the three right circular cones. As IRSs are cost-effective compared with conventional mmWave devices, they can be massively installed with minimal effort. We can expect that IRSs assisted mmWave with a large number of delicately placed IRSs is capable to guarantee $|\mathcal{N}| \geq 3$ unblocked links with high probability.

In practice,  estimation error $\varepsilon_\eta$ cannot be zero.
To estimate the 3-D position $\mathbf{p} = (x,y,z)^T$, least square criterion is adopted, i.e.,
\begin{align} \label{OptRev1}
\begin{split}
&\min_{\mathbf{p}} \;\;  \xi_{\phi}(\mathbf{p}) \triangleq \sum_{\eta\in \mathcal{N}}  \left(\hat{\phi}_\eta  - \phi_\eta(\mathbf{p}) \right)^2   \\
& s.t. \;\;\;\;   \mathbf{p} \in \mathcal{S}
\end{split}
\end{align}
where $\mathcal{S}$ is the position range of indoor MT, e.g., the 3-D space of lecture hall. As the objective function $\xi_{\phi}(\mathbf{p})$ is non-convex, it is non-trivial to derive the analytical solution to the problem. Fortunately, Taylor-series estimation method is capable to effectively solve a large class of position-location problems\cite{foy1976position}. Starting with a rough initial guess, the Taylor-series estimation method iteratively improves its guess at each step by determining the local linear least-sum-squared-error correction\cite{foy1976position}. In AoD based positioning, with an initial position guess $\hat{\mathbf{p}}$,  the following approximation can be obtained through Taylor series expansion by neglecting $m$-th order terms ($m\geq 2$), i.e.,
\begin{align}
 \phi_\eta(\mathbf{p}) \approx \phi_\eta(\hat{\mathbf{p}}) + (\mathbf{p} - \hat{\mathbf{p}})^T \frac{\partial \phi_\eta({\mathbf{p}})}{\partial {\mathbf{p}}}\bigg|_{\mathbf{p} = \hat{\mathbf{p}}} \label{Approx1}
\end{align}
where the first order derivative is denoted as
\begin{align}
 \frac{\partial \phi_\eta(\mathbf{p})}{\partial \mathbf{p}}  =  \frac{\| \mathbf{p} - \mathbf{p}_\eta\|_2 \mathbf{e}_\eta -  (\mathbf{p} - \mathbf{p}_\eta)^T\mathbf{e}_\eta \frac{ \mathbf{p} - \mathbf{p}_\eta}{\|  \mathbf{p} - \mathbf{p}_\eta\|_2} }{\| \mathbf{p} - \mathbf{p}_\eta\|_2^2}
\end{align}
Substituting \eqref{Approx1} into \eqref{MeaAoD}, we have
\begin{align}
\hat{\phi}_\eta - \phi_\eta(\hat{\mathbf{p}}) \approx  \frac{\partial \phi_\eta({\mathbf{p}})}{\partial {\mathbf{p}}^T}\big|_{\mathbf{p} = \hat{\mathbf{p}}} (\mathbf{p} - \hat{\mathbf{p}}) + \varepsilon_\eta, \;\; \eta \in \mathcal{N}
\end{align}
Its matrix form is written as
\begin{align}
\Delta_{\boldsymbol{\phi}} \approx \mathbf{A}^T \Delta_{\mathbf{p}}  + \boldsymbol{\varepsilon} \label{Approx}
\end{align}
where  $ \Delta_{\mathbf{p}}=   \mathbf{p}-\hat{\mathbf{p}}$, $\boldsymbol{\varepsilon} = [\varepsilon_1, \cdots, \varepsilon_{|\mathcal{N}|}]^T$, and
\begin{subequations}
\begin{align}
 \Delta_{\boldsymbol{\phi}} = [ \hat{\phi}_1 - \phi_1(\hat{\mathbf{p}}), \cdots, \hat{\phi}_{|\mathcal{N}|} - \phi_{|\mathcal{N}|}(\hat{\mathbf{p}}) ]^T  \label{MatrixA1}\\
\mathbf{A}= \left[\frac{\partial \phi_1({\mathbf{p}})}{\partial {\mathbf{p}}}\big|_{\mathbf{p} = \hat{\mathbf{p}}}, \cdots, \frac{\partial \phi_{|\mathcal{N}|}({\mathbf{p}})}{\partial {\mathbf{p}}}\big|_{\mathbf{p} = \hat{\mathbf{p}}}  \right]  \label{MatrixA2}
\end{align}
\end{subequations}
On the basis of \eqref{Approx}, the Taylor series method for AoD based positioning is summarized in Algorithm 1.


\begin{figure}[tp]{
\begin{center}{\includegraphics[ height=7cm]{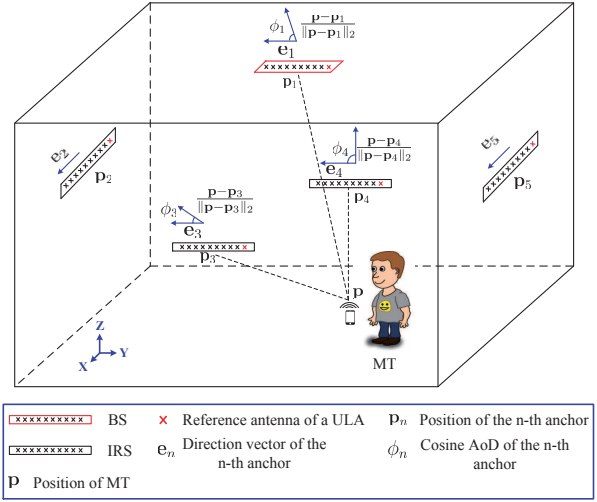}}
\caption{A typical scenario of IRS assisted mmWave communications}\label{ScenarioPositioning1}
\end{center}}
\end{figure}

\begin{algorithm}[t]

    \caption{Taylor Series Method For AoD Based Positioning}
           \textbf{Initialization}: Generate an initial guess of MT position $\hat{\mathbf{p}}$ \\
           \textbf{Input}: The estimate of cosine AoDs of a selected set of reliable links, i.e., $\hat{\phi}_\eta, (\eta \in \mathcal{N})$, positions  of  anchors $\mathbf{p}_\eta, (\eta \in \mathcal{N})$, directions of ULA on the anchors  $\mathbf{e}_\eta, (\eta \in \mathcal{N})$, and iteration stopping parameter $\epsilon$. \\
         \textbf{\emph{Repeat}} {
         \\
               \quad \emph{1.}  With the given $\hat{\mathbf{p}}$, generate $\phi_\eta(\hat{\mathbf{p}}), (\eta \in \mathcal{N})$ according to Eq. \eqref{MatrixA1}  and $\mathbf{A}$  according to Eq.\eqref{MatrixA2}.
          \\
               \quad \emph{2.}    Find the least square estimate of  $\Delta_{\mathbf{p}}$, i.e.,
            \begin{align}
            \hat{\Delta}_{\mathbf{p}} = (\mathbf{A}\mathbf{A}^T)^{-1}\mathbf{A}\Delta_{\boldsymbol{\phi}} \notag
            \end{align}
          \\
             \quad \emph{3.}      Update $\hat{\mathbf{p}}$,  i.e., $ \hat{\mathbf{p}} \leftarrow  \hat{\mathbf{p}}  + \hat{\Delta}_{\mathbf{p}}$.}\\
      \textbf{\emph{Until}} $\|\hat{\Delta}_{\mathbf{p}}\|_2 <\epsilon$.
 \end{algorithm}

\subsubsection{Reliable Link Set $\mathcal{N}$}

An intuitive method to construct the set of reliable links is to select $|\mathcal{N}|$ links with the  $|\mathcal{N}|$ smallest $\varpi_\eta$ to avoid unreliable AoDs resulted from blockage and Error Type 2 of joint AoA, AoD estimation. However, it is non-trivial to determine the exact value of $|\mathcal{N}|$.
Although $|\mathcal{N}| = 3$ anchors are theoretically sufficient to yield the position of MT in the ideal noiseless case, more anchors are desirable in practice for positioning algorithm to enhance the accuracy of  position estimation.

To utilize as many reliable anchors as possible, the following strategy is proposed to iteratively construct the reliable link set $\mathcal{N}$.
Firstly, we sort the anchors in ascending order according to residual signal power ratio $\varpi_\eta$. Then, starting from $|\mathcal{N}| = 3$ anchors, we iteratively increase the number of anchors used for positioning in Algorithm 1, and by the end of each iteration, we calculate {  the cost $ {\frac{\xi_{\phi}(\mathbf{p})}{|\mathcal{N}|}}$, where $\xi_{\phi}(\mathbf{p})$ is the squared error of least square method in Eq. \eqref{OptRev1} and $|\mathcal{N}|$ is the number of selected anchors}. Finally, we select the output corresponding to the largest $|\mathcal{N}|$ that satisfies { $ {\frac{\xi_{\phi}(\mathbf{p})}{|\mathcal{N}|}} \leq \xi_{th}$ } as the estimated position of MT, where $\xi_{th}$ is a preset threshold\footnotemark.

\footnotetext{An appropriate $\xi_{th}$ can be obtained by carrying out a great number of Monte Carlo experiments offline. In our numerical experiment, we find that $\sqrt{\xi_{th}} = 0.005$ results in a good performance. }

\subsection{Parameter Estimation With The Aid of MT Position }
With the estimated position $\hat{\mathbf{p}}$, channel parameters can be refined according to the geometric relationship.

\subsubsection{AoD Refinement}
With $\hat{\mathbf{p}}$, AoD estimation is updated by
\begin{align}
{\phi}_\eta^\star = \frac{(\hat{\mathbf{p}} - \mathbf{p}_\eta)^T\mathbf{e}_\eta}{\| \hat{\mathbf{p}} - \mathbf{p}_\eta\|_2},\;\; \eta\in \{1, 2, \cdots, N_{IRS}+1 \}
\end{align}

\subsubsection{AoA Refinement}

To estimate AoA, the direction of ULA in MT's side  is essential. Therefore, we firstly find the least square estimate of $\mathbf{e}_{MT}$ by solving the following optimization problem.
\begin{align} \label{Direction}
\begin{split}
& \min_{\mathbf{e}_{MT}}\;\; \xi_{\theta}(\mathbf{e}_{MT}) \triangleq \sum_{\eta \in \mathcal{N}} \left(\frac{(\hat{\mathbf{p}} - \mathbf{p}_{\eta})^T \mathbf{e}_{MT}}{\| \hat{\mathbf{p}} - \mathbf{p}_\eta\|_2} \ominus  \hat{\theta}_\eta \right)^2   \\
&s.t. \;\;\;\; \| \mathbf{e}_{MT} \|_2 = 1
\end{split}
\end{align}
Note that $\mathcal{N}$ can be derived in the iterative process according to Section V. A. 3.

The objective function of \eqref{Direction} can be rewritten in matrix form as
\begin{align}
\xi_{\theta}(\mathbf{e}_{MT}) = \| \mathbf{P}^T \mathbf{e}_{MT} \ominus \hat{\boldsymbol{\theta}} \|_2^2
\end{align}
where  $\mathbf{P} = \left[\frac{\hat{\mathbf{p}} - \mathbf{p}_{ {\eta}_1} }{\| \hat{\mathbf{p}} - \mathbf{p}_{{\eta}_{1}}\|_2} \cdots \frac{\hat{\mathbf{p}} - \mathbf{p}_{ {\eta}_{|\mathcal{N}|}}}{\| \hat{\mathbf{p}} - \mathbf{p}_{{\eta}_{|\mathcal{N}|}}\|_2} \right]$, $\hat{\boldsymbol{\theta}} = [\hat{\theta}_{{\eta}_1}, \cdots, \hat{\theta}_{{\eta}_{|\mathcal{N}|}}]^T$
and $\mathcal{N} = \{\eta_1, \cdots, \eta_{|\mathcal{N}|} \}$.
The optimization problem can be solved via projected gradient descent method  \cite{ProjGradient}, in which we iteratively update $\mathbf{e}_{MT}$ as follows.
\begin{align}
\begin{split}
&\mathbf{d}_{MT,i+1} =  \mathbf{e}_{MT,i}  - \lambda \frac{\partial \xi_{\theta}( \mathbf{e}_{MT}) }{\partial\mathbf{e}_{MT}}\Big|_{\mathbf{e}_{MT} = \mathbf{e}_{MT,i}} \\
&\mathbf{e}_{MT,i+1} = \frac{\mathbf{d}_{MT,i+1}}{\|\mathbf{d}_{MT,i+1} \|_2}
\end{split}
\end{align}
where $\lambda$ is step size and $\frac{\partial \xi_{\theta}(\mathbf{e}_{MT}) }{\partial \mathbf{e}_{MT}} = \mathbf{P} \left( \mathbf{P}^T \mathbf{e}_{MT} \ominus \hat{\boldsymbol{\theta}} \right)$.

Finally, with $\hat{\mathbf{e}}_{MT}$  yielded by projected gradient descent method, AoA estimation is updated by
\begin{align}
{\theta}_\eta^\star = \frac{(\hat{\mathbf{p}} - \mathbf{p}_\eta)^T\hat{\mathbf{e}}_{MT}}{\| \hat{\mathbf{p}} - \mathbf{p}_\eta\|_2}
\end{align}

\subsubsection{Estimation of Blockage}

As a prerequisite of our proposed blockage estimation method, we firstly introduce the estimation of $\delta_\eta$, which is  dependent on the values of $(\theta_\eta, \phi_\eta)$. Note that the parameter estimate obtained in Section IV by ML estimation is under the assumption that $\zeta_\eta = 1$, while it is probable that $\zeta_\eta = 0$ in fact. It would be misleading in the estimation of $\delta_\eta$ by directly substituting $(\hat\theta_\eta, \hat\phi_\eta)$ into \eqref{OptDelta}. Therefore, we will use the estimates of AoA and AoD  refined by position to assist the estimation of $\delta_\eta$ and $\zeta_\eta$, as they are cross verified by multiple anchors and are thus more reliable.

Substituting $({\theta}_\eta^\star,  {\phi}_\eta^\star)$ into \eqref{OptDelta}, we have
\begin{align}
 {\delta}_\eta^\star  & =   \frac{\mathbf{b}^H( \theta_\eta^\star,  \phi_\eta^\star)  \mathbf{D}^H  \mathbf{y}}{\| \mathbf{D} \mathbf{b}( \theta_\eta^\star,  \phi_\eta^\star) \|_2^2} \notag \\
& =   \frac{\zeta_\eta \delta_\eta\mathbf{b}^H(\theta_\eta^\star,  \phi_\eta^\star)  \mathbf{D}^H  \mathbf{D} \mathbf{b}(\theta_\eta,\phi_\eta) + {\mathbf{b}^H(\theta_\eta^\star,  \phi_\eta^\star)  \mathbf{D}^H  \mathbf{n}} }{\| \mathbf{D} \mathbf{b}(\theta^\star, \phi^\star) \|_2^2} \notag \\
& = \zeta_\eta \delta_\eta f({\theta}_\eta^\star, {\phi}_\eta^\star) + \bar{n}
\end{align}
where $f(\theta_\eta^\star, \phi_\eta^\star) \triangleq  \frac{\mathbf{b}^H(\theta_\eta^\star, \phi_\eta^\star)  \mathbf{D}^H  \mathbf{D} \mathbf{b}(\theta_\eta,\phi_\eta)}{\| \mathbf{D} \mathbf{b}(\theta_\eta^\star, \phi_\eta^\star) \|_2^2}$, $\bar{n} \sim \mathcal{CN}(0,  \sigma^2_{\bar{n}})$, and $\sigma^2_{\bar{n}} = \frac{\sigma^2_{\bar{\mathbf{w}}} + \sigma^2_{\boldsymbol{\nu}}}{\| \mathbf{D} \mathbf{b}(\theta_\eta^\star, \phi_\eta^\star) \|_2^2}$   (or $\sigma^2_{\bar{n}} =\frac{\sigma^2_{\bar{\mathbf{w}}} + \sigma^2_{\boldsymbol{\nu}_1} + \sigma^2_{\boldsymbol{\nu}_2}}{\| \mathbf{D} \mathbf{b}(\theta_\eta^\star, \phi_\eta^\star) \|_2^2}$). Thus, we have
\begin{align}
{\delta}_\eta^\star  = \left \{ \begin{array}{cc}
                             \delta_\eta f({\theta}_\eta^\star, {\phi}_\eta^\star)  + \bar{n}, & \;\; \zeta_\eta=1 \\
                           \bar{n}, & \;\; \zeta_\eta = 0
                         \end{array} \right.
\end{align}

Theoretically, with the knowledge of $\delta_\eta$, $f({\theta}_\eta^\star, {\phi}_\eta^\star)$ and $\sigma^2_{\bar{n}}$, the decision of $\zeta_\eta$ can be made by comparing the probabilities of ${\delta}_\eta^\star$ conditioned on $\zeta_\eta = 0 $ and $\zeta_\eta = 1$. However, accurate estimation of $f({\theta}_\eta^\star, {\phi}_\eta^\star)$ and $\sigma^2_{\bar{n}}$ is challenging in practice. With respect to $\delta_\eta$, its amplitude $|\delta_\eta|$ is estimable from the distance of MT, while its phase cannot be accurately estimated from the distance, as it is very sensitive to distance estimation error and may be affected by random initial phase of local oscillator in transmitter side.

Alternatively, a heuristic method is proposed to decide blockage indicator by comparing the pathloss estimated from $(\theta^\star_\eta, \phi^\star_\eta)$ and  pathloss estimated from $\hat{\mathbf{p}}$, i.e.,
\begin{align} \label{ZetaDec}
\left|10 \log_{10} \frac{1}{|{{\delta}}_\eta^\star|^2} - 10 \log_{10} \frac{1}{|{\delta}_\eta(\hat{\mathbf{p}})|^2} \right| \begin{array}{c}
                                                            \zeta_\eta^\star = 1 \\
                                                            \lesseqgtr \\
                                                            \zeta_\eta^\star = 0
                                                          \end{array}
 PL_{th}
\end{align}
where
\begin{align}
&|{\delta}_\eta(\hat{\mathbf{p}})| = \left\{ \begin{array}{c}
                                  \left|\frac{\sqrt{P_{Tx}} \lambda e^{-j 2\pi d_{BM}}}{4\pi d_{BM}}\right|,\qquad \qquad  \;\; \eta = 1  \notag \\
                                  \left|\frac{ \sqrt{\xi P_{Tx} N_{B}}\lambda e^{-j 2\pi (d_{BR_\eta} + d_{R_\eta M}) }}{4\pi(d_{BR_\eta} + d_{R_\eta M})}\right|,  \eta = 2,\cdots, N_{IRS + 1}
                               \end{array} \right. \notag
\end{align}
BS/AP to MT distance $d_{BM}$ and IRS to MT distance $d_{R_\eta M}$ are attainable from $\hat{\mathbf{p}}$, and $ PL_{th}$ is the preset threshold of pathloss distance (In numerical simulations, we set $PL_{th} = 6$ dB ).

\section{Numerical Results}
In this section, we numerically study the performance of the proposed joint beam training and positioning scheme for IRSs assisted mmWave MIMO.

\subsection{Settings of Numerical Experiment}

\begin{table}[tp] \centering
\noindent
\caption{Simulation Parameters} \setlength{\belowcaptionskip}{0cm} \label{T1}
\begin{tabular}{ll}
\hline
   \textbf{Parameter}    &    \textbf{Value}   \\
\hline
    Operating frequency     &     $28$ GHz     \\
    Noise power      &     $-84$ dBm      \\
    Position of IRSs     &    \tabincell{c}{$(5,\;-10,\; 3.5)$, $(5,\;10,\; 3.5)$,\quad \\  $(0,\;-10,\; 3.5)$,  $(0,\;10,\; 3.5)$, \\ $(-5,\;-10,\; 3.5)$,  $(-5,\;10,\; 3.5)$, \\ $(-10,\;5,\; 3.5)$,  $(10,\;5,\; 3.5)$, \\  $(-10,\;0,\; 3.5)$,  $(10,\;0,\; 3.5)$, \\  $(-10,\;-5,\; 3.5)$,  $(10,\;-5,\; 3.5)$} \\
    Position of BS/AP     &    $(0,\;0,\;5)$       \\
    Direction of  IRSs' ULA     &    \tabincell{c}{$(0,\;0,\; 1)$, $(1,\;0,\; 0)$, $(0,\;1,\; 0)$,  \\  $(0,\;1,\; 0)$, $(0,\;0,\; 1)$, $(1,\;0,\; 0)$,  \\ $(0,\;1,\; 0)$,  $(0,\;1,\; 0)$,  $(0,\;0,\; 1)$, \\ $(1,\;0,\; 0)$, $(0,\;1,\; 0)$,  $(0,\;1,\; 0)$     \;}       \\
    Direction of   BS/AP's ULA     &     $(\frac{\sqrt{2}}{2},\;\frac{\sqrt{2}}{2},\; 0)$      \\
    Reflection loss $-10\log_{10} \xi $    &   \;   $13$ dB   \\
    Size of obstacles     &       $0.6  \times 0.4 \times  1.7$  meters  \\
    Altitude of MT     &       $[1.2,\; 1.4]$  meters  \\
    Number of users     &        $ 20, 50, 100  $\\
    Number of NLoS paths      &        $0$ {  (only in \figref{MSEComp})},  $4$\\
    {  Number of antennas in BS/AP ($N_B$)} & $16$  \\
    {  Number of antennas in MT ($N_M$)}  & $16$ \\
    {  Number of reflectors in IRS ($N_R$)} & $16$\\
\hline
\end{tabular}
\end{table}

We assume that  IRSs-assisted mmWave MIMO system is deployed in an indoor scenario, e.g., lecture hall, and the length, width and height of which are $20$ meters, $20$ meters and $5$ meters, respectively.  The rest system parameters are listed in Table \ref{T1}. For simplicity, we assume that AoA, AoD of  NLoS paths follow uniform distribution, i.e., $\theta_{BM,l}, \phi_{BM,l} \sim U(0, 2\pi), l=2,...,L$, and path coefficient follows complex Gaussian distribution, i.e.,  $\delta_{l} \sim \mathcal{CN}(0, \sigma_l^2), l=2,...,L $ and $10 \log_{10} \frac{\delta_{1}^2}{\sigma_l^2} = 20$ dB. We model user (MT holder) as a cube with its length, width and height being $0.6$m, $0.4$m and $1.7$m, respectively. We denote  position of the MT held by user as $(x,y,z)$, where  $x, y, z$ follow  uniform distribution, i.e.,  $x,y \sim U(-10,10)$ and $z \sim U(1.2,1.4)$. Users are uniformly distributed in the lecture hall under the non-overlapping constraint. For a typical MT, the other MT holders are its potential obstacles, and thus the blockage probability increases with user density.

\subsection{Relationship Between User Density and Blockage Probability}
To gain insights into the relationship between user density and blockage probability, \figref{PlotObstaclesBlockage} is presented where there are $12$ IRSs deployed, which means a total of $13$ LoS/VLoS links are available. From the \figref{PlotObstaclesBlockage}, we can see that when the number of MTs is $20$, more than $50\%$ of channel realizations experience no link blockage, the largest number of blocked links is $4$, and the percentage of which is less than $5\%$; when the number of MTs is $50$, more than $80\%$ of channel realizations experience less than $3$ blocked links, the largest number of blocked links is $7$, and the percentage of which is less than $1\%$; when the number of MTs is $100$,  more than $80\%$ of channel realizations experience less than $5$ blocked links,  the largest number of blocked links is $9$, and the percentage of which is almost negligible. Note that when there exists at least $1$ unblocked link, uninterrupted communication over mmWave band can be guaranteed, and when there exist at least $3$ unblocked links, positioning algorithm can be performed to locate MT and meanwhile enhance parameter estimation.


\begin{figure}[tp]{
\begin{center}{\includegraphics[ height=3cm]{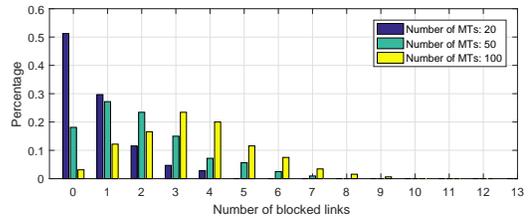}}
\caption{Blockage rate with different user densities}\label{PlotObstaclesBlockage}
\end{center}}
\end{figure}

\subsection{Performance of Beam Training with Random Beamforming}

{ 
As the performance of  joint beam training and positioning is fundamentally determined by the decomposed Sub-problem 1 for BS/AP-MT link and Sub-problem 2 for BS/AP-IRS-MT links, whose unified signal model is Eq. \eqref{UniModel}, we start numerical evaluation from the sub-problems, i.e., the beam training scheme with random beamforming proposed in Section IV. The blockage indicator $\zeta$ of Eq. \eqref{UniModel} is set as $\zeta = 1$, and the random variable $\mathbf{n}= \mathbf{w} + \boldsymbol{\nu}$, where $\mathbf{w}$ is the noise term and $\boldsymbol{\nu}$ is the interference term. The noise term  $\mathbf{w} \sim \mathcal{CN}(\mathbf{0}, \sigma_\mathbf{w}^2 \mathbf{I})$ and $\sigma_\mathbf{w}^2$ is $-86$ dBm according to Table II. The interference term $\boldsymbol{\nu}$ is propagated via NLoS paths, and its entries are represented in Eq. \eqref{MeasureBS} for BS/AP-MT link and in Eq. \eqref{Sysmodel2} for BS/AP-IRS-MT links. A notable difference between  $\boldsymbol{\nu}$ and $\mathbf{w}$ is that the power of $\boldsymbol{\nu}$ is proportional to transmit power. Since Sub-problem 1 and Sub-problem 2 are mathematically equivalent, we carry out the numerical study of  beam training with random beamforming in BS/AP-MT link in this subsection.}

 \begin{figure}[tp]{
\begin{center}{\includegraphics[ height=7.8cm]{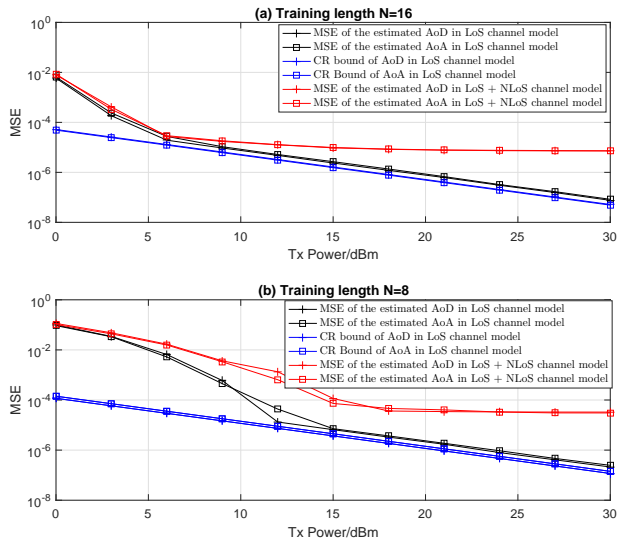}}
\caption{MSE performance of AoA/AoD estimated by random beamforming based beam training in LoS channel model and LoS + NLoS  channel model, where NLoS path number is $4$, training length is $N=8, 16$.  }\label{MSEComp}
\end{center}}
\end{figure}

{ In \figref{MSEComp}, we use mean squared error (MSE) of the estimated AoA/AoD as the performance metric, which is defined as $MSE(\hat\theta) \triangleq \mathbb{E}\left((\hat{\theta} \ominus \theta)^2\right), MSE(\hat\phi) \triangleq \mathbb{E}\left((\hat{\phi} \ominus \phi)^2\right)$, where $(\hat{\theta}, \hat{\phi})$ are the estimated AoA and AoD of the LoS path, and $({\theta}, {\phi})$ are the exact values of AoA and AoD of the LoS path.
The proposed beam training scheme is characterized by two steps, namely random beamforming and ML estimation. Random beamforming is performed to measure mmWave channel,  and ML estimation is performed to estimate AoA and AoD of the LoS path based on channel measurements. To study the accuracy of ML estimator, we use Cram\'{e}r-Rao bound\footnotemark (CRB) in the ideal LoS channel (where $\boldsymbol{\nu} = \mathbf{0}$) as the benchmark.}  It can be seen from \figref{MSEComp}(a) that, when the training length is $N = 16$, from $0$ dBm to $6$ dBm the empirical MSE of both AoA and AoD in LoS mmWave channel is significantly higher than CRB, but the performance gap gradually turns to be marginal from $6$ dBm to above. It indicates that, from $0$ dBm to $6$ dBm ML estimation of $(\theta, \phi)$ experiences Error Type 2 as mentioned in Section IV. C, in which the estimated AoA and AoD pair are far apart from their authentic values, and from $6$ dBm to above only Error Type 1 happens, in which the estimation error is mild and tightly lower bounded by CR bound. It validates the effectiveness of ML estimator in relative high SNR regimes. In practice, NLoS path's impacts on beam training cannot be overlooked.
In the numerical simulation of beam training in LoS + NLoS mmWave channel, we set the number of NLoS paths as $4$. As can be seen from \figref{MSEComp}(a) that, from $0$ dBm to $6$ dBm the empirical MSE of  AoA and AoD in LoS + NLoS channel is slightly worse than that in LoS channel, which indicates that noise is the main detrimental factor. From $9$ dBm to above, the MSE curves turn to be flat, and this is because the impact of NLoS path, namely $\boldsymbol{\nu}$, does not diminish over SNR. A notable point is that MSE from $9$ dBm to above is around $10^{-5}$, which is satisfactorily accurate. To study the impact of training length, MSE performance comparison is also performed when $N=8$ in \figref{MSEComp}(b). A remarkable difference from $N=16$ case is that the flat curves of empirical MSE start from $18$ dBm, and the values of which are around $10^{-4}$, which indicates that the impact of noise  in $N=8$ case is more significant than $N=16$ case and thus verifies the benefits of increasing training length.

\footnotetext{Since the estimation of $(\theta, \phi)$ is part of the joint estimation of $(\delta, \theta, \phi)$, CRBs of $ \theta$ and $\phi$ are obtained as the last two diagonal elements of the inverse of Fisher information matrix  w.r.t. $(\delta, \theta, \phi)$. The detailed derivation of CRB is omitted, as it follows the standard procedure.}

\begin{figure}[tp]{
\begin{center}
\begin{minipage}[!h]{0.8\linewidth}
\includegraphics[ width=1.1\textwidth]{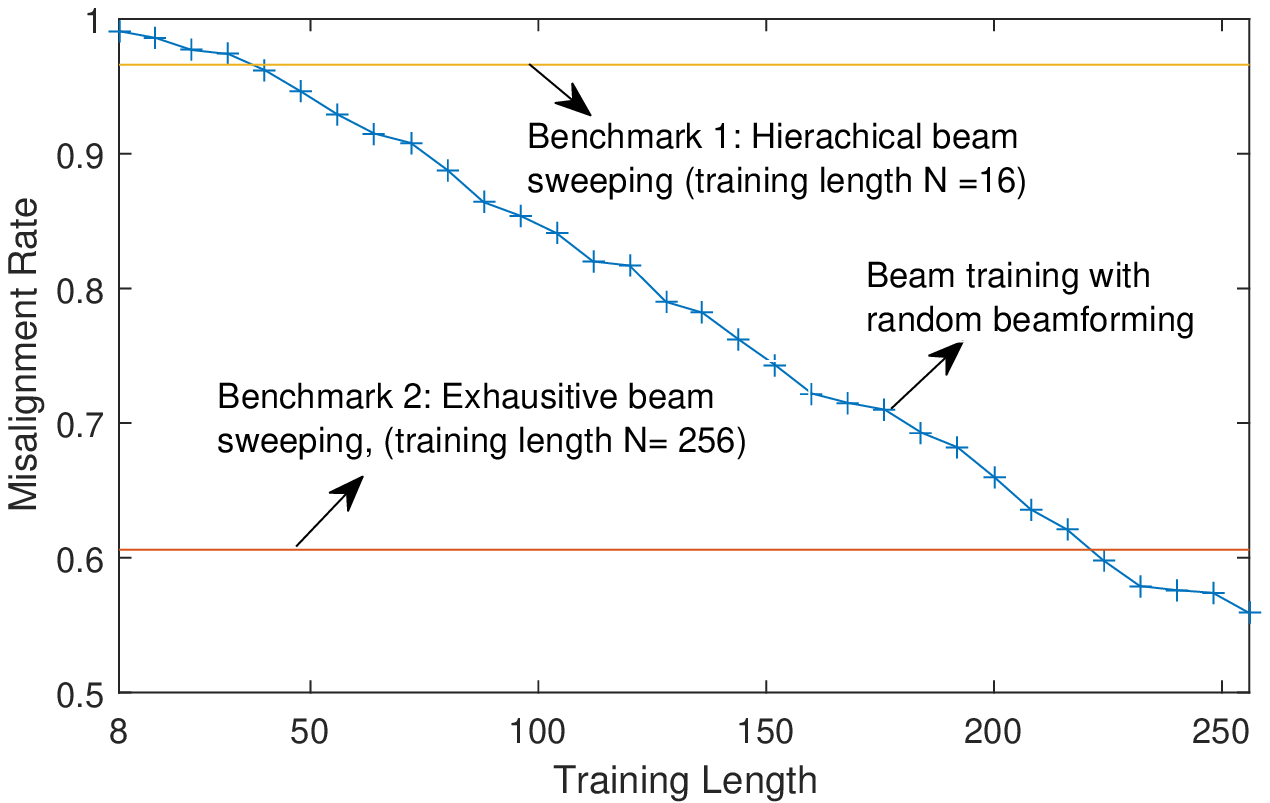}
\subcaption{When transmit power is $-20$ dBm }
\label{N4}
\end{minipage}
\begin{minipage}[!h]{0.8\linewidth}
\centering
\includegraphics[ width=1.1\textwidth]{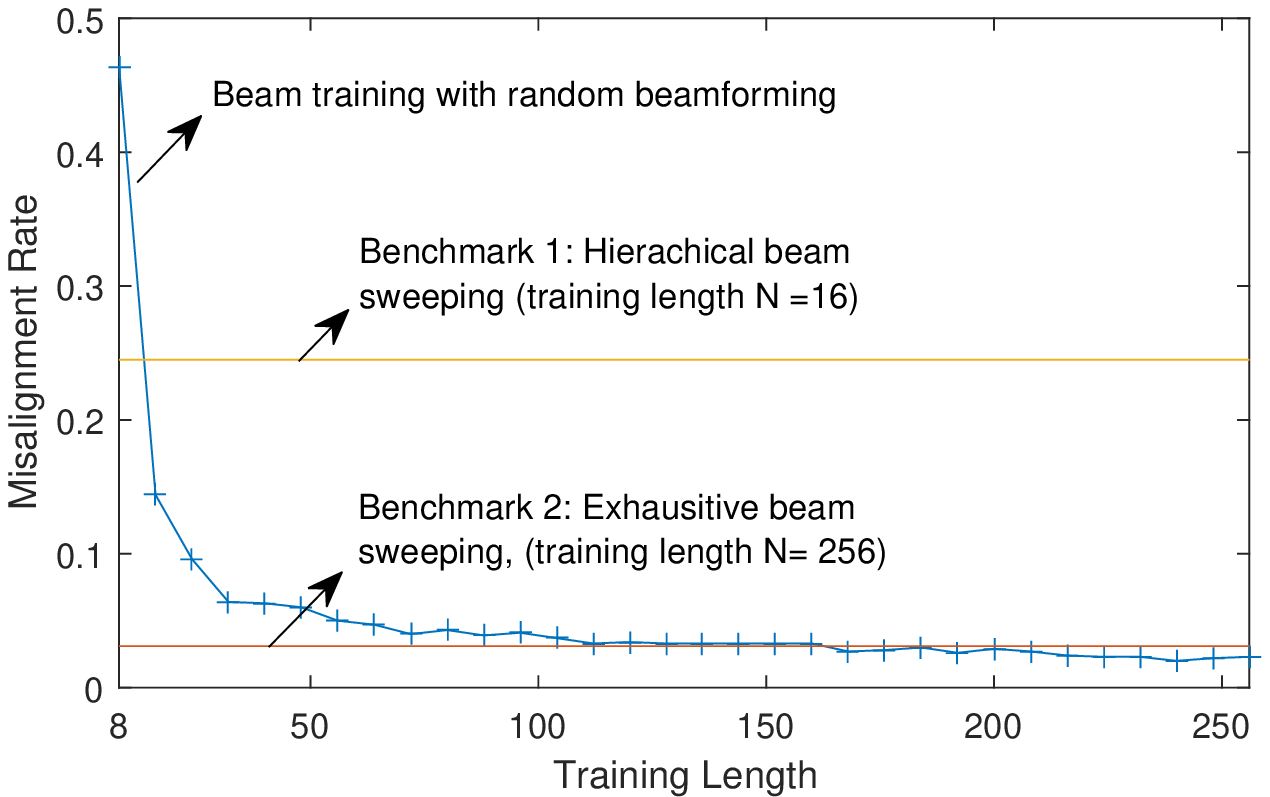}
\subcaption{When transmit power is $0$ dBm }
\label{N4}
\end{minipage}
\end{center}
\caption{Performance (misalignment rate) comparison between the proposed beam training with random beamforming and the existing beam training schemes with directional beamforming }\label{MisalignmentComp}}
\end{figure}

{ 
In \figref{MisalignmentComp}, we compare the performance of the proposed random beamforming based beam training scheme with the existing directional beamforming based beam training schemes \cite{giordani2018tutorial, Mmwave5, xiao2016hierarchical}. Directional beamforming is used for beam training in a more straightforward way than random beamforming, in which the candidate beams pairs are explored through exhaustive/hierachical beam sweeping, and then the strongest beam pair is selected based on the received power of the candidates. Directional beams are selected from a pre-configured finite set corresponding to quantized angles, e.g., discrete Fourier Transform (DFT) codebook. To compare the performance of random beamforming and directional beamforming in beam training, we use beam misalignment rate as the performance metric, which measures the probability that beam training fails to find the strongest beam pair. For random beamforming based beam training, we select the best beam pair by quantizing the estimated AoA/AoD  to its nearest codeword. Two types of directional beamforming techniques are used as the benchmarks, namely, exhaustive beam sweeping \cite{giordani2018tutorial} and hierarchical beam sweeping \cite{xiao2016hierarchical, Mmwave5}. Exhaustive beam sweeping explores all the possible beam pairs and its training length is $N=N_BN_M = 256$; Hierarchical beam sweeping iteratively narrows down the direction search region and results in logarithmic training length, i.e., $N = 4\log_2 \min(N_B, N_M) + 2 \log_2 \frac{\max(N_B, N_M)}{\min(N_B, N_M)}  = 16$. By contrast, random beamforming is flexible with training length. In the simulation, we set the training length of random beamforming as $N = {8,16,\cdots, 248, 256}$ to investigate the impact of training length. We compare the performance of random beamforming based beam training with directional beamforming based beam training at two SNR levels, i.e., $P_{Tx}= -20$ dBm, $0$ dBm, in LoS + NLoS channel model. From \figref{MisalignmentComp}(a), it can be seen that, when $P_{Tx}= -20$ dBm, the misalignment rate of exhaustive beam sweeping is $0.606$, and the misalignment rate of hierarchical beam sweeping is $0.966$. The beam misalignment rate of random beamforming is $0.991$ when training length is $N=8$,  and it decreases over training length and turns to be $0.559$ when training length is $N=256$. It verifies the conclusion of Theorem 2 and indicates that random beamforming with an appropriate training length could achieve better performance than directional beamforming. From \figref{MisalignmentComp}(b), it can be seen that, when $P_{Tx}= 0$ dBm, the misalignment rate of exhaustive beam sweeping is $0.031$, and the misalignment rate of hierarchical beam sweeping is $0.245$. As for random beamforming, the performance improvement over training length becomes more significant. Specifically, the misalignment rate is $0.464$ when training length is $N=8$ and sharply decreases to $0.144$ when $N=16$, and finally it converges to  $0.023$ when $N = 256$. It is noteworthy that the performance enhancement brought by increasing training length is marginal from $N=32$. Therefore, the training length of random beamforming can be set adaptively according to SNR condition to achieve a satisfactory performance with moderate training cost.
}

\subsection{Performance of Joint Beam Training and Positioning for IRSs Assisted MmWave Communications }

\begin{figure}[tp]{
\begin{center}{\includegraphics[ height=4.2cm]{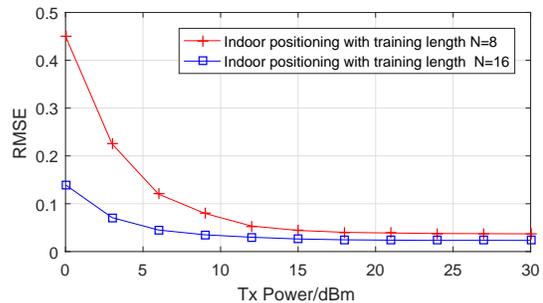}}
\caption{Accuracy of indoor positioning}\label{PositionMSE}
\end{center}}
\end{figure}

{ 
In this subsection, we study the performance of joint beam training and positioning for IRSs assisted mmWave communications. The configurations of IRSs, BS and MT, which determine the path gain, AoA and AoD of the LoS path, are given in Table. \ref{T1}. In addition, we set the number of users as $100$, which determines the blockage indicator, and we also set the number of NLoS as $4$.
}

In \figref{PositionMSE}, the accuracy of indoor positioning of IRSs assisted mmWave MIMO is studied in terms of root mean squared error (RMSE). When the training length is $N=16$ for each LoS/VLoS path, RMSE is $0.13$ meter at $0$ dBm, and converges to $0.02$ meter from $15$ dBm to $30$ dBm, which indicates that, with the aid of IRSs, mmWave MIMO achieves centimeter accuracy in indoor scenario. When the training length is $N=8$ for each LoS/VLoS path, RMSE is $0.45$ meter at $0$ dBm, and converges to $0.04$ meter from $15$ dBm to $30$ dBm. Considering the reduced training length, the accuracy limit of $0.04$ meter for $N=8$ case in high SNR regimes is acceptable. However,
the positioning accuracy of $N=8$ case is not satisfying in low SNR regimes. Through case analysis, we find that the correlation between residual ratio $\varpi_{\eta}$ and the  accuracy of $(\hat{\theta}_{\eta}, \hat{\phi}_{\eta})$ is weakened by the increased level of noise and the reduced training length. In other words, a small $\varpi_{\eta}$ may misleadingly correspond to an unreliable anchor node, and thus results in inaccurate estimate of position.  To improve the accuracy, a more sophisticated positioning algorithm that iteratively sorts the reliability will be developed in the future.

\begin{figure}[tp]{
\begin{center}{\includegraphics[ height=4.2cm]{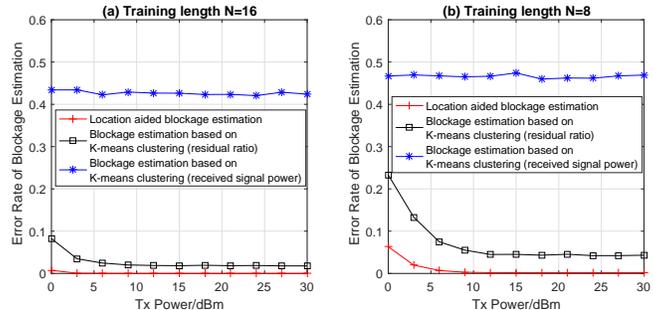}}
\caption{Error rate of blockage estimation}\label{BLK}
\end{center}}
\end{figure}

In \figref{BLK}, the error rate of blockage estimation is studied. For the purpose of comparison, two methods are adopted as benchmarks, which are (1) received power based blockage estimation and (2) residual ratio based blockage estimation. For (1), it is straightforward that unblocked links have significantly higher received signal level than that of blocked links. However, as power level is an absolute quantity, without the prior knowledge such as the likely range of received power, it is possible to mistake the unblocked link between MT and faraway anchor as a blocked link. In contrast, residual ratio in (2) is a relative quantity, which is not dependant on the likely range of received power. However, the optimal threshold that is essential for blockage estimation is unavailable either. Therefore, we adopt the K-means clustering method to  partition the $13$ observations into $2$ clusters, i.e., blocked links and unblocked links. When the training length is $N=16$, we can see from the figure that position aided blockage estimation is slightly erroneous merely at $0$ dBm and becomes errorless when transmit power increases. With respect to the benchmark methods, although the estimation accuracy of residual ratio based K-means clustering method is worse than position aided blockage estimation, its error rate is below $0.1$, which is acceptable. By contrast, the estimation error rate of received power based K-means clustering method is nearly $0.5$, which indicates that the estimation is almost random. When the training length reduces to $N=8$, the superiority of position aided blockage estimation is more remarkable, and this is owing to the cross-validation mechanism enabled by location information.

\begin{figure}[tp]{
\begin{center}{\includegraphics[ height=9.8cm]{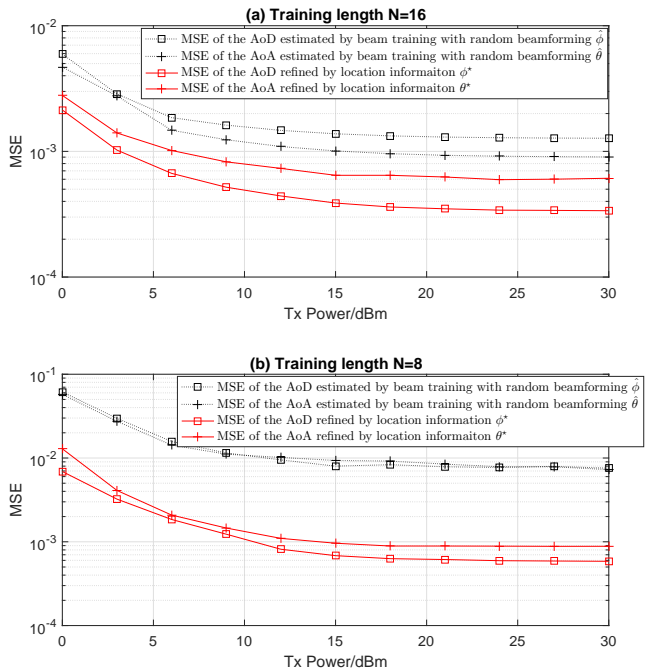}}
\caption{MSE performance of AoA/AoD refined by location information }\label{AoAAoDpos}
\end{center}}
\end{figure}

In \figref{AoAAoDpos}, MSE performance evaluation of AoA/AoD refined by location information is performed. To this end, we intentionally filter out the blocked links, and reserve AoA/AoD estimate of the unblocked links. As can be seen that AoA/AoD refined by location information is more accurate than AoA/AoD estimated by beam training with random beamforming. This is because location information is derived by multiple anchors, and AoA/AoD refinement according to geometric relationship means that the estimation is cross verified. It is noteworthy that the performance enhancement is more significant when the training length is $N=8$ for each LoS/VLoS path, from which we find the potential to reduce training length of beam training with the aid of location information. Another notable point is that AoA refined by location information is always worse than AoD refined by location information. This is because the direction vector $\hat{\mathbf{e}}_{MT}$ is derived from estimation in  \eqref{Direction}, while the direction vectors of anchors   $\mathbf{e}_{\eta}$ are well known.

\section{Conclusion}
In this paper, beam training for IRSs assisted mmWave communications is studied. By breaking down beam training for IRSs assisted mmWave MIMO into several mathematically equivalent sub-problems, we perform random beamforming and maximum likelihood estimation  to derive the optimal beam of BS/AP and MT  and the optimal reflection pattern of IRSs. Then, by sorting the reliability of the estimated AoA, AoD paris, we propose an iterative positioning algorithm to acquire the position of MT, and with which we are able to cross verify and enhance the estimation of AoA and AoD, and accurately predict link blockage. Numerical results show the superiority of our proposed beam training scheme and verify the performance gain brought by location information.
\begin{appendices}

\section{Partial derivatives of $g({\theta, \phi})$}
The derivative of $g({\theta, \phi})$ with respect to $\theta$ is
\begin{align} \label{Gradient1}
&\frac{\partial g({\theta, \phi}) }{\partial \theta}  \notag \\
= & \frac{\frac{\partial \mathbf{b}^H (\theta, \phi)\mathbf{D}^H\mathbf{y}\mathbf{y}^H\mathbf{D}\mathbf{b}(\theta, \phi)}{\partial \theta}}{ \mathbf{b}^H(\theta, \phi)\mathbf{D}^H \mathbf{D}\mathbf{b}(\theta, \phi) } - \notag \\
 &\frac{\mathbf{b}^H (\theta, \phi)\mathbf{D}^H\mathbf{y}\mathbf{y}^H\mathbf{D}\mathbf{b}(\theta, \phi)}{ \left( \mathbf{b}^H(\theta, \phi)\mathbf{D}^H \mathbf{D}\mathbf{b}(\theta, \phi) \right)^2} \frac{\partial \mathbf{b}^H (\theta, \phi)\mathbf{D}^H \mathbf{D}\mathbf{b}(\theta, \phi)}{\partial \theta} \notag \\
 = & 2 Re \left( \frac{ \mathbf{b}^H (\theta, \phi)\mathbf{D}^H\mathbf{y}\mathbf{y}^H\mathbf{D}\frac{\partial \mathbf{b}(\theta, \phi)}{\partial \theta}}{ \mathbf{b}^H(\theta, \phi)\mathbf{D}^H \mathbf{D}\mathbf{b}(\theta, \phi) } - \notag \right. \\
 &\left. \frac{\mathbf{b}^H (\theta, \phi)\mathbf{D}^H\mathbf{y}\mathbf{y}^H\mathbf{D}\mathbf{b}(\theta, \phi)}{ \left( \mathbf{b}^H(\theta, \phi)\mathbf{D}^H \mathbf{D}\mathbf{b}(\theta, \phi) \right)^2} { \mathbf{b}^H (\theta, \phi)\mathbf{D}^H \mathbf{D} \frac{\partial \mathbf{b}(\theta, \phi)}{\partial \theta}} \right)
\end{align}
where $\frac{\partial \mathbf{b} (\theta, \phi) }{\partial \theta}=  vec\left(  (\mathbf{a}_{Rx}(\theta) \odot  \boldsymbol{\vartheta}_{Rx} ) \mathbf{a}_{Tx}^H(\phi) \right) $ and  $\boldsymbol{\vartheta}_{Rx} = \left[0, j\pi, \cdots, j\pi(N_r -1)\right]^T $.
Similarly, the derivative of $g({\theta, \phi})$ with respect to $\phi$ is
\begin{align}  \label{Gradient2}
&\frac{\partial g({\theta, \phi}) }{\partial \phi}
 =  2 Re \left( \frac{ \mathbf{b}^H (\theta, \phi)\mathbf{D}^H\mathbf{y}\mathbf{y}^H\mathbf{D}\frac{\partial \mathbf{b}(\theta, \phi)}{\partial \phi}}{ \mathbf{b}^H(\theta, \phi)\mathbf{D}^H \mathbf{D}\mathbf{b}(\theta, \phi) } - \notag \right. \\
 &\left. \frac{\mathbf{b}^H (\theta, \phi)\mathbf{D}^H\mathbf{y}\mathbf{y}^H\mathbf{D}\mathbf{b}(\theta, \phi)}{ \left( \mathbf{b}^H(\theta, \phi)\mathbf{D}^H \mathbf{D}\mathbf{b}(\theta, \phi) \right)^2} { \mathbf{b}^H (\theta, \phi)\mathbf{D}^H \mathbf{D} \frac{\partial \mathbf{b}(\theta, \phi)}{\partial \phi}} \right)
\end{align}
where $\frac{\partial \mathbf{b} (\theta, \phi) }{\partial \phi} = vec\left(  \mathbf{a}_{Rx}(\theta) (\mathbf{a}_{Tx}(\phi) \odot  \boldsymbol{\vartheta}_{Tx})^H     \right)$ and  $\boldsymbol{\vartheta}_{Tx} = \left[0, j\pi, \cdots, j\pi(N_t -1)\right]^T$.

\section{Proof of Theorem 1}

In the noiseless scenario where $\mathbf{y} = \mathbf{D} \mathbf{b} (\theta, \phi)$,
according to Cauchy-Schwarz inequality, we have
\begin{align} \label{CS}
 &\left\| \frac{\mathbf{b}^H (\widetilde{\theta}, \widetilde{\phi})\mathbf{D}^H}{\|  \mathbf{D} \mathbf{b} (\widetilde{\theta}, \widetilde{\phi})\|_2} \mathbf{D} \mathbf{b} (\theta, \phi)\right\|_2 \leq  \| \mathbf{D} \mathbf{b} (\theta, \phi)\|_2
\end{align}
Then, the proof of Eq. \eqref{UniEst} is reduced to prove that
\begin{align}
 &  \left\| \frac{\mathbf{b}^H (\widetilde{\theta}, \widetilde{\phi})\mathbf{D}^H}{\|  \mathbf{D} \mathbf{b} (\widetilde{\theta}, \widetilde{\phi})\|_2} \mathbf{D} \mathbf{b} (\theta, \phi)\right\|_2 \neq   \| \mathbf{D} \mathbf{b} (\theta, \phi)\|_2  \label{ProEQ}
\end{align}
namely $\mathbf{D} \mathbf{b} (\theta, \phi) \neq \mu \mathbf{D} \mathbf{b} (\widetilde{\theta}, \widetilde{\phi}),  \;\;\forall \mu \in \mathbb{C}, \forall (\theta, \phi) \neq (\widetilde{\theta}, \widetilde{\phi})$, which is mathematically equivalent to  Eq. \eqref{UniRresentation}.

\section{Proof of Theorem 2}
The PEP is written as
\begin{align}
&\qquad Pe\left((\theta, \phi) \rightarrow (\widetilde{\theta}, \widetilde{\phi})\right) \notag \\
= &   Pr\left(\left\|\frac{\mathbf{b}^H (\theta, \phi)\mathbf{D}^H}{\|  \mathbf{D} \mathbf{b} (\theta, \phi)\|_2} \mathbf{y}\right\|_2^2 < \left\|\frac{\mathbf{b}^H (\widetilde{\theta}, \widetilde{\phi})\mathbf{D}^H}{\|  \mathbf{D} \mathbf{b} (\widetilde{\theta}, \widetilde{\phi})\|_2} \mathbf{y}\right\|_2^2   \right)    \notag \\
= & Pr \Bigg(  {- \frac{|\mathbf{b}^H ( {\theta},  {\phi})\mathbf{D}^H\mathbf{n}|^2 }{\|  \mathbf{D} \mathbf{b} ( {\theta},  {\phi})\|_2}  + \frac{|\mathbf{b}^H (\widetilde{\theta}, \widetilde{\phi})\mathbf{D}^H\mathbf{n}|^2}{\|  \mathbf{D} \mathbf{b} (\widetilde{\theta}, \widetilde{\phi})\|_2}}  \notag \\
& \qquad     - 2 \Re\left\{  {\delta \mathbf{n}^H \mathbf{D}\mathbf{b}(\theta, \phi)}  \right\} \notag \\
 & \qquad + 2 \Re\left\{  \frac{\delta \mathbf{n}^H \mathbf{D}\mathbf{b}(\widetilde{\theta}, \widetilde{\phi}) \mathbf{b}^H(\widetilde{\theta}, \widetilde{\phi}) \mathbf{D}^H \mathbf{D}\mathbf{b}({\theta}, {\phi})}{\|  \mathbf{D} \mathbf{b} (\widetilde{\theta}, \widetilde{\phi})\|_2^2} \right\}  \notag \\
& \qquad > \|\delta \mathbf{D} \mathbf{b} (\theta, \phi)\|_2^2 -   \frac{|\delta \mathbf{b}^H (\widetilde{\theta}, \widetilde{\phi})\mathbf{D}^H\mathbf{D}\mathbf{b}(\theta,\phi)|^2}{\|  \mathbf{D} \mathbf{b} (\widetilde{\theta}, \widetilde{\phi})\|_2^2}     \Bigg) \notag \\
\approx & Pr \Bigg(   N_1  > \|\delta \mathbf{D} \mathbf{b} (\theta, \phi)\|_2^2 -   \frac{|\delta \mathbf{b}^H (\widetilde{\theta}, \widetilde{\phi})\mathbf{D}^H\mathbf{D}\mathbf{b}(\theta,\phi)|^2}{\|  \mathbf{D} \mathbf{b} (\widetilde{\theta}, \widetilde{\phi})\|_2^2}     \Bigg) \label{EqNoise}
\end{align}
where
\begin{align}
&\qquad \qquad N_1 = \notag \\
&2\Re\left\{ -  {\delta \mathbf{n}^H \mathbf{D}\mathbf{b}(\theta, \phi)} +  \frac{\delta \mathbf{n}^H \mathbf{D}\mathbf{b}(\widetilde{\theta}, \widetilde{\phi}) \mathbf{b}^H(\widetilde{\theta}, \widetilde{\phi}) \mathbf{D}^H \mathbf{D}\mathbf{b}({\theta}, {\phi})}{\|  \mathbf{D} \mathbf{b} (\widetilde{\theta}, \widetilde{\phi})\|_2^2}  \right\} \notag
\end{align}
and
$\Re\{\cdot\}$ is the real part of a complex number. Eq. \eqref{EqNoise} is
obtained by neglecting the component $- \frac{|\mathbf{b}^H ( {\theta},  {\phi})\mathbf{D}^H\mathbf{n}|^2 }{\|  \mathbf{D} \mathbf{b} ( {\theta},  {\phi})\|_2}  + \frac{|\mathbf{b}^H (\widetilde{\theta}, \widetilde{\phi})\mathbf{D}^H\mathbf{n}|^2}{\|  \mathbf{D} \mathbf{b} (\widetilde{\theta}, \widetilde{\phi})\|_2}$ in high SNR regime. Since $N_1$ is a Gaussian random variable, we
have
\begin{align}
& \qquad \qquad N_1 \sim \notag\\
&\mathcal{N} \left( 0, 2\sigma^2 |\delta|^2 \left( \|  \mathbf{D} \mathbf{b} (\theta, \phi)\|_2^2  -  \frac{|\mathbf{b}^H (\widetilde{\theta}, \widetilde{\phi})\mathbf{D}^H  \mathbf{D} \mathbf{b}(\theta, \phi)|^2 }{\|  \mathbf{D} \mathbf{b} (\widetilde{\theta}, \widetilde{\phi}) \|_2^2}  \right)    \right) \notag
\end{align}
According to the definition of Q function, \eqref{Th1} is obtained.

\section{Proof of Proposition 1}
Firstly, we write the expression of $d^2(\mathbf{D}_n, \theta, \phi, \widetilde{\theta}, \widetilde{\phi})$ as
\begin{align}
 & \qquad d^2(\mathbf{D}_n, \theta, \phi, \widetilde{\theta}, \widetilde{\phi})  \notag \\
 =& \left\|  \mathbf{D}_n \mathbf{b} (\theta, \phi)\right\|_2^2  -   \frac{|\mathbf{b}^H (\widetilde{\theta}, \widetilde{\phi})\mathbf{D}_n^H\mathbf{D}_n \mathbf{b}(\theta, \phi)|^2}{\|  \mathbf{D}_n \mathbf{b} (\widetilde{\theta}, \widetilde{\phi})\|_2^2}  \notag \\
 =&   \mathbf{b}^H(\theta,\phi)\mathbf{D}_{n-1}^H  \mathbf{D}_{n-1} \mathbf{b}(\theta,\phi) + \mathbf{b}^H(\theta,\phi)\mathbf{d}_n\mathbf{d}_n^H  \mathbf{b}(\theta,\phi)  -   \notag \\
& \frac{\left|\mathbf{b}^H(\widetilde{\theta},\widetilde{\phi})\mathbf{D}_{n-1}^H  \mathbf{D}_{n-1} \mathbf{b}(\theta,\phi) + \mathbf{b}^H(\widetilde{\theta},\widetilde{\phi})\mathbf{d}_n\mathbf{d}_n^H  \mathbf{b}(\theta,\phi)\right|^2}{\mathbf{b}^H(\widetilde{\theta},\widetilde{\phi})\mathbf{D}_{n-1}^H  \mathbf{D}_{n-1} \mathbf{b}(\widetilde{\theta},\widetilde{\phi}) + \mathbf{b}^H(\widetilde{\theta},\widetilde{\phi})\mathbf{d}_n\mathbf{d}_n^H  \mathbf{b}(\widetilde{\theta},\widetilde{\phi})} \notag
 \end{align}
Thus
\begin{align}
&d^2(\mathbf{D}_n, \theta, \phi, \widetilde{\theta}, \widetilde{\phi})  -  d^2(\mathbf{D}_{n-1}, \theta, \phi, \widetilde{\theta}, \widetilde{\phi}) \notag \\
 =& \mathbf{b}^H(\theta,\phi)\mathbf{d}_n\mathbf{d}_n^H  \mathbf{b}(\theta,\phi)  + \frac{ |\mathbf{b}^H(\widetilde{\theta},\widetilde{\phi})\mathbf{D}_{n-1}^H  \mathbf{D}_{n-1} \mathbf{b}(\theta,\phi) |^2}{\mathbf{b}^H(\widetilde{\theta},\widetilde{\phi})\mathbf{D}_{n-1}^H  \mathbf{D}_{n-1} \mathbf{b}(\widetilde{\theta},\widetilde{\phi})  } \notag \\ &-\frac{|\mathbf{b}^H(\widetilde{\theta},\widetilde{\phi})\mathbf{D}_{n-1}^H  \mathbf{D}_{n-1} \mathbf{b}(\theta,\phi) + \mathbf{b}^H(\widetilde{\theta},\widetilde{\phi})\mathbf{d}_n\mathbf{d}_n^H  \mathbf{b}(\theta,\phi)|^2}{\mathbf{b}^H(\widetilde{\theta},\widetilde{\phi})\mathbf{D}_{n-1}^H  \mathbf{D}_{n-1} \mathbf{b}(\widetilde{\theta},\widetilde{\phi}) + \mathbf{b}^H(\widetilde{\theta},\widetilde{\phi})\mathbf{d}_n\mathbf{d}_n^H  \mathbf{b}(\widetilde{\theta},\widetilde{\phi})} \notag
\end{align}
For the purpose of conciseness, let
\begin{align}
 \check{a} &= \mathbf{b}^H(\widetilde{\theta},\widetilde{\phi})\mathbf{d}_n\mathbf{d}_n^H  \mathbf{b}(\widetilde{\theta},\widetilde{\phi});\notag \\
 \check{b} &=  \mathbf{b}^H(\widetilde{\theta},\widetilde{\phi})\mathbf{d}_n\mathbf{d}_n^H  \mathbf{b}( {\theta}, {\phi});\notag \\
 \check{c} &= \mathbf{b}^H(\widetilde{\theta},\widetilde{\phi})\mathbf{D}_{n-1}^H  \mathbf{D}_{n-1} \mathbf{b}(\widetilde{\theta},\widetilde{\phi});\notag \\
 \check{d} &=  \mathbf{b}^H(\widetilde{\theta},\widetilde{\phi})\mathbf{D}_{n-1}^H  \mathbf{D}_{n-1} \mathbf{b}(\theta,\phi). \notag
\end{align}
As $\mathbf{b}^H(\widetilde{\theta},\widetilde{\phi})\mathbf{d}_n$ and $\mathbf{d}_n^H  \mathbf{b}( {\theta}, {\phi})$ are numbers, rather than vectors, we have
\begin{align}
\mathbf{b}^H(\theta,\phi)\mathbf{d}_n\mathbf{d}_n^H  \mathbf{b}(\theta,\phi) &= \frac{|\mathbf{b}^H(\widetilde{\theta},\widetilde{\phi})\mathbf{d}_n\mathbf{d}_n^H  \mathbf{b}( {\theta}, {\phi})|^2}{\mathbf{b}^H(\widetilde{\theta},\widetilde{\phi})\mathbf{d}_n\mathbf{d}_n^H  \mathbf{b}(\widetilde{\theta},\widetilde{\phi})} =\frac{ |\check{b}|^2}{ \check{a}} \notag
\end{align}
Then,
\begin{align}
&d^2(\mathbf{D}_{n}, \theta, \phi, \widetilde{\theta}, \widetilde{\phi})  - d^2(\mathbf{D}_{n-1}, \theta, \phi, \widetilde{\theta}, \widetilde{\phi}) \notag \\
 =&\frac{ |\check{b}|^2}{ \check{a}} + \frac{|\check{d}|^2}{\check{c}} - \frac{|\check{b} +\check{d}|^2 }{ \check{a}+  \check{c}} \notag \\
 =& \frac{  |\check{b}|^2 \check{c}( \check{a} + \check{c})  +  |\check{d}|^2 \check{a}( \check{a} +  \check{c}) - \check{a} \check{c}|\check{b} +\check{d}|^2}{ \check{a} \check{c}( \check{a} +  \check{c})} \notag \\
  =& \frac{  |\check{b}|^2 \check{c}( \check{a} +  \check{c})  +  |\check{d}|^2 \check{a}( \check{a} +  \check{c}) - \check{a} \check{c}|\check{b}|^2 -   \check{a} \check{c} |\check{d}|^2  - 2 \check{a} \check{c} Re\{\check{b}^* \check{d}\}   }{ \check{a} \check{c}( \check{a} +  \check{c})} \notag \\
    =& \frac{  |\check{b}|^2 \check{c}^2  +  |\check{d}|^2 \check{a}^2     - 2 \check{a} \check{c} Re\{\check{b}^* \check{d}\}   }{ \check{a} \check{c}( \check{a} +  \check{c})} \notag \\
    =& \frac{ |\check{a}\check{d} - \check{b} \check{c}|^2  }{ \check{a} \check{c}( \check{a} +  \check{c})}  \geq 0 \notag
\end{align}
and equality holds when $\check{a}\check{d} - \check{b} \check{c} = 0$, namely,
\begin{align}
\frac{\mathbf{b}^H(\widetilde{\theta},\widetilde{\phi})\mathbf{d}_n\mathbf{d}_n^H  \mathbf{b}(\widetilde{\theta},\widetilde{\phi})}{ \mathbf{b}^H(\widetilde{\theta},\widetilde{\phi})\mathbf{d}_n\mathbf{d}_n^H  \mathbf{b}( {\theta}, {\phi}) }  = & \frac{\mathbf{b}^H(\widetilde{\theta},\widetilde{\phi})\mathbf{D}_{n-1}^H  \mathbf{D}_{n-1} \mathbf{b}(\widetilde{\theta},\widetilde{\phi})}{\mathbf{b}^H(\widetilde{\theta},\widetilde{\phi})\mathbf{D}_{n-1}^H  \mathbf{D}_{n-1} \mathbf{b}(\theta,\phi)  }
\end{align}

\end{appendices}

\bibliographystyle{IEEEtran}%

\bibliography{bibfile}

\end{document}